\documentclass[12pt]{iopart}

\usepackage{iopams}
\usepackage{amsthm,dsfont}
\usepackage{graphicx}
\usepackage[small, bf]{caption}

\newcommand{\EW}[1]{\mathbb{E}\left[#1\right]}
\newcommand{\Var}[1]{\mathbb{V}\!ar\left[#1\right]}
\newcommand{\Prob}[1]{\mathbb{P}\left[#1\right]}
\newcommand{\Probc}[2]{\mathbb{P}_{#1}\left[#2\right]}
\newcommand{\eqref}[1]{(\ref{#1})}
\newcommand{\Xspace}{\mathbb{X}}
\newcommand{\sol}{f}
\newcommand{\altsol}{g}
\newcommand{\manifold}{D}
\newcommand{\half}{{\scriptsize{\frac{1}{2}}}}

\newcommand{\N}{\mathbb{N}}	% Natuerliche Zahlen
	
\newcommand{\R}{\mathbb{R}}	% Reelle Zahlen
\newcommand{\paren}[1]{\left( #1 \right)}
\newcommand{\norm}[1]{\left\| #1 \right\|}
\newcommand{\argmin}{\,\mathrm{argmin}\,}
\renewcommand{\ker}{\,\mathrm{ker}\,}
\newcommand{\ran}{\,\mathrm{ran}\,}
\newcommand{\grad}{\,\mathrm{grad}\,}
\newcommand{\Div}{\,\mathrm{div}\,}

\newcommand{\calR}{\mathcal{R}}
\newcommand{\calS}{\mathcal{S}}
\newcommand{\calT}{\mathcal{T}}

\newcommand{\frakB}{\mathfrak{B}}
\newcommand{\uF}{u}
\newcommand{\vTest}{v}
\newcommand{\shift}{\tau}
\newcommand{\mymatrix}[2]{\left(\begin{array}{#1} #2\end{array}\right)}

\newcommand{\KL}{\mathrm{KL}}

\newcommand{\Opg}{F}

\newcommand{\bh}{\mathbf{h}}
\newcommand{\bfX}{\mathbf{X}}
\newcommand{\bfY}{\mathbf{Y}}
\newcommand{\bfW}{\mathbf{W}}
\renewcommand{\bmu}{\boldsymbol{\mu}}

\newcommand{\muhatk}{\widehat\bmu_k}

\newcommand{\solhat}{\widehat\sol}
\newcommand{\solhatk}{\widehat\sol_k}
\newcommand{\solhatkm}{\widehat\sol_{k-1}}

\newtheorem{theorem}{Theorem}
\newtheorem{corollary}[theorem]{Corollary}
\newtheorem{lemma}[theorem]{Lemma}
\newtheorem{proposition}[theorem]{Proposition}
\newtheorem{example}[theorem]{Example}

\theoremstyle{definition}

\theoremstyle{remark}
\newtheorem{remark}{Remark}

\begin{document}

\title[Parameter identification in SDEs]{On parameter identification in stochastic differential equations by penalized maximum likelihood}

\author{Fabian Dunker and Thorsten Hohage}

\address{Institut f\"ur Numerische und Angewandte Mathematik, Georg-August Universit\"at G\"ottingen

Lotzestr. 16--18, 37083 G\"ottingen, Germany}
\ead{dunker@math.uni-goettingen.de}
\begin{abstract}
In this paper we present nonparametric estimators for coefficients in stochastic 
differential equation if the data are described by independent, identically distributed 
random variables. The problem is formulated as  a nonlinear ill-posed operator equation 
with a deterministic forward operator described by the Fokker-Planck equation. 
We derive convergence rates of the risk for penalized maximum likelihood estimators 
with convex penalty terms and for Newton-type methods. The assumptions of our 
general convergence results are verified for estimation of the drift coefficient. 
The advantages of log-likelihood compared to quadratic data fidelity 
terms are demonstrated in Monte-Carlo simulations. 
\end{abstract}

%Uncomment for PACS numbers title message
%\pacs{02.30.Zz, 05.40.-a, 02.30.Jr}
% Keywords required only for MST, PB, PMB, PM, JOA, JOB? 
%\vspace{2pc}
%\noindent{\it Keywords}: Article preparation, IOP journals
% Uncomment for Submitted to journal title message
%\submitto{\JPA}
% Comment out if separate title page not required
%\maketitle

\section{Introduction}\label{sec:intro}
Many dynamical processes in physics, social sciences and economics can be modeled by systems of stochastic differential equations
\begin{equation}\label{eq:sde}
d\bfX_t = \bmu(t, \bfX_t) dt + \sigma(t,\bfX_t) d \bfW_t.
\end{equation}
Here $t \in [0,T]$ with $T > 0$ is interpreted as time, $\bfX_t$ is a family of random variables with values in 
$\mathbb{R}^d$, and $\bfW_t$ is a standard Wiener process in $\R^d$. 
The function $\bmu: [0,T] \times \R^d  \to \R^d$ is called drift coefficient while 
$\sigma : [0,T] \times \R^d \to \R^{d\times d}$ is the volatility or diffusion.
Observations of the process give values of one or more paths $(\bfX_t)_{t \ge 0}$ at one or many times $t$. 
In many applications there is an interest to estimate the drift or the diffusion either non-parametrically or parametrically to gain a better understanding of the modeled process. 

In this paper we consider the particular case where $\bmu$ and $\sigma$ are independent of $t$, $\sigma$ is known while $\bmu$ should be estimated. Let us describe two kinds of observations suitable for our approach: 
\begin{enumerate}
\item An ensemble of independent paths $\bfX_t^{(i)}$, $i=1,\dots,n$ is observed at a fixed time $t=T$. I.e.\ the observations are the random variables $\bfY_i=\bfX_T^{(i)}$. The starting points of the paths $\bfX_0^{(i)}$ are assumed to be sampled from a known distribution $u_0$ 

\item We observe only one path of a strictly stationary, ergodic process at equidistant times. 
I.e.\ our observations are $\bfY_i= \bfX_{(i+i_0)\Delta t}$ for $i=1,\dots,n$ and $i_0>0$.  
\end{enumerate}

%In this paper we will analyze a new nonparametric estimator for the drift using methods of parameter identification in deterministic 
%partial differential equations. However, the method can be adapted to estimate the diffusion as well. 
Our approach to the problem is based on the Fokker-Planck equation, also called forward Kolmogorov equation. 
Assume $\bfX_t$ has a sufficiently smooth density $\uF(t,\cdot)$ for all 
$t \in [0,T]$.%
%, and $\mathbf{n}$ is the outer normal vector on $\manifold$. 
Then \eqref{eq:sde} holds true if and only if $\uF$ solves 
the initial value problem 
\begin{equation}\label{eq:fokker-planck}
\eqalign{
\frac{\partial}{\partial t} \uF = \Div \left( - \bmu \uF + \half  \sigma\sigma^{\top} \grad \uF \right)\\
%2  \uF (\bmu\cdot \mathbf{n}) = \sigma\sigma^{\top} \grad \uF \cdot \mathbf{n} \qquad \mbox{on } [0,T] \times \partial\manifold\\
\uF(0,\cdot) = \uF_0% \qquad \mbox{for all } x \in \manifold
}
\end{equation}
(see e.g.\ \cite{Risken:89}). Hence, we can define the deterministic 
coefficient-to-solution operator $F(\bmu):=\uF(T,\cdot)$. This operator is nonlinear.
%$\Op(\bmu) = \uF|_M$ which maps the drift $\bmu$ to the density $\uF|_M$ solving 
%\eqref{eq:fokker-planck} restricted to an observation set $M\subset \manifold\times (0,\infty)$. 

In case of an ergodic process with $\bmu$, $\sigma$ not depending on 
$t$, solutions to eq.~\eqref{eq:fokker-planck} tend to a 
stationary solution as $t\to\infty$ which solves the elliptic equation
\begin{equation}\label{eq:fokker-planck_stationary_unbounded}
\eqalign{
0 = \Div\left(-\bmu \uF + \half   \sigma\sigma^{\top} \grad \uF\right)\\
%2  \uF  (\bmu\cdot \mathbf{n}) = \left(\sigma\sigma^{\top} \grad \uF\right) \cdot \mathbf{n} \qquad \mbox{on} \; \partial\manifold\\
\int \uF(x) dx = 1.
}
\end{equation}
Here the coefficient-to-solution operator is defined by $F(\bmu) := \uF$. The operator $F$ and its 
properties will be discussed in Section \ref{sec:Fokker-Planck}. 
%Conditions under which $\Op$ and $F$ are well-defined and differentiable will be discussed in Section \ref{sec:Fokker-Planck}.

We will derive convergence results for general operators $F$ with values in a set of probability densities. The unknown of the inverse problem will be denoted by $\sol$ in this general case. In the setting above 
we have $\sol=\bmu$, but in other applications $\sol=\sigma$ or 
$\sol=(\bmu,\sigma)$. If parametric estimation is preferred over non-parametric estimation, $\sol$ can be a parameter in a model of 
$\bmu$ or $\sigma$. 
Suppose that $\sol^{\dag}$ is the exact solution and $\uF^{\dag}:=F(\sol^{\dag})$ 
the corresponding probability density. 
We assume that the observed data are described by independent random variables 
$\bfY_1,\dots,\bfY_n$ each of which has probability density $\uF^{\dagger}$.  
Note that equidistant observations $\bfY_i= {\bf X}_{(i+i_0)\Delta t}$ of one path are actually 
not independent. Therefore, our results apply immediately  %with $\uF$ as limiting distribution ${\bf X}_t$ for $t\to \infty$
only to the first scenario where an ensemble of independent paths is observed. 
In the second scenario additional information is contained in the order of the 
data $\bfY_i$ which will be neglected here. This is justified 
if $\Delta_t$ is so large that the dependence of $\bfY_i$ and $\bfY_{i+1}$ is neglectible 
or if no information on the order is available. 
%However, in the second scenario the ergodicity assumption ensures that averages over observations of one path at many times asymptotically coincide with averages over an ensemble of paths at some large time $T$. This justifies the application of our 
%method in the second case. 

Our estimator follows the idea to seek an estimator $\solhat$ which maximizes the likelihood of the given observations $\bfY_i=y_i$. 
It is convenient to describes these observations by the empirical measure 
\begin{equation}\label{eq:empirical_meas}
\Phi_n := \frac{1}{n}\sum_{i=1}^n \delta_{y_i}.
\end{equation} 
%As $\EW{\int_{\manifold}\varphi d\Phi_n} = \int_{\manifold}\varphi(x)\uF^{\dagger}(x)\,dx$, 
%the allows in particular to recover such linear functionals of $\uF^{\dagger}$ in expectation.
Since $\Probc{\uF}{y_1,\dots,y_n}=\prod_{i=1}^n \uF(y_i)$, the negative log-likelihood is given by 
\begin{equation}\label{eq:S0}
\calS_0(\Phi_n, \uF) 
= -\frac{1}{n}\ln \Probc{\uF}{y_1,\dots,y_n}
= - \frac{1}{n}\sum_{i=1}^n \ln \uF(y_i) 
= - \int \ln(\uF)\, d\Phi_n.
\end{equation}
%(If $\calS(\Phi_n, \uF)$ is minimized over a set of densities
%$\uF$ we may obviously omit the term $\int \uF(x)\,dx$ as it is identically $1$. However, this is not the case in the context 
%of Newton-type methods where $F$ is approximated by a linear operator.) 
Due to ill-posedness a simple maximum likelihood estimator, i.e.\ a minimizer of $\calS_0\left(\Phi_n, F(\sol)\right)$ 
over $\sol$ in some convex set $\mathfrak{B}$, is unstable. Therefore, we have to regularize. 
In the (generalized) Tikhonov regularization one adds a penalty 
term $\calR:\frakB \to \R\cup\{\infty\}$, which we assume to be
convex, lower semi-continuous, and not identically $\infty$.  
It is weighted by a regularization parameter $\alpha>0$: 
\begin{equation}\label{eq:Tikh}
\solhat_{\alpha} \in \argmin_{\sol\in\frakB}\left[
\calS\left(\Phi_n; F(\sol)\right) + \alpha \calR(\sol) \right].
\end{equation}
Due to the non-linearity of $F$ this is in general a non-convex minimization 
problem even though $\calS\left(\Phi_n; \cdot\right)$ 
and $\calR$ are convex. An alternative is to locally approximate 
$F$ around a current iterate by its Fr\'echet derivative 
$F'[\solhatk]$. This yields the iteratively 
regularized Newton method
\begin{equation}\label{eq:Gauss-Newton}
\solhatk \in \argmin_{\sol\in \frakB}
\left[\calS\left(\Phi_n; F'[\solhatkm](\sol - \solhatkm) 
+ F(\solhatkm)\right) + \alpha_k \calR(\sol) \right].
\end{equation}
Here $(\alpha_k)$ is a sequence of positive regularization parameters converging monotonically to $0$ for increasing $k$ such that 
$\alpha_k/\alpha_{k+1}$ remains bounded. To assure well-posedness 
of these optimization problems and to analyze convergence, it is 
often necessary to ''regularize'' the data fidelity term 
$\calS$. This is of particular importance when $\uF$ is negative on 
a set of positive measure which implies $\calS(\Phi_n,\uF)= \infty$. 
A further discussion is contained in Section~\ref{sec:general_convergence}. 

 % The same algorithm can be stated for $\Op$ instead of $F$ by taking the empirical measure $\Phi_n$ for the position of all observed paths at any point in time 
% $t$ and minimize $\calS\left(\hat\Phi_t, \Op(\solhat)(t,\cdot)\right)$ simultaneously for all $t$. If we have a non-random design in $t$, 
% an appropriate way is to minimize $\left\|\calS\left(\hat\Phi_t, \Op(\solhat)\right)\right\|_{L^2[0,T]}$. 
% %If the observations comprise more than one time $t$, the average of $\calS\left(\hat\Phi_t, \Op(\solhat)_t\right)$ over all 
% %available $t$ can be minimized to use the observations more efficiently.
%

%If the negative log likelihood functional $\calS$ is replaced by a quadratic data fidelity term and possibly $\Phi_n$ by some 
%density estimator, we obtain a classical iteratively regularized Gau{\ss}-Newton method. This method is of natural if the data are normally distributed.

All known convergence rate results for regularization methods involving $F'$ under 
source conditions weaker than $\sol^{\dag}\in \ran(F'[\sol^{\dagger}]^*)$ 
require additional assumptions on $F'$ such as the tangential cone condition
\begin{equation}\label{eq:L2tcc}
\|F(\altsol) - F(\sol) - F'[\sol](\altsol-\sol)\|_{L^2} 
\leq \eta \|F(\altsol) - F(\sol)\|_{L^2}.
\end{equation}
For KL-type data fidelity terms a related formulation \eqref{eq:KL-tcc} suggested 
recently in \cite{HW:13} is required. For parameter identification problems 
for which $D(F)$ and $\ran(F)$ are function spaces over different domains these 
conditions are typically very difficult to verify, but if the domains coincide 
the $L^2$ tangential cone condition has been shown for a number of problems 
(see e.g.\ \cite{HNS:95,DZ:13}). To the best of our knowledge for drift estimation 
in the stationary Fokker-Planck equation \eqref{eq:fokker-planck_stationary_unbounded} 
both the $L^2$-version and in particular the KL-version of the tangential cone 
condition are unknown so far, and we will prove them below. 

%For the implementation of the iteratively regularized Newton method the operator and its derivative have to be evaluated in 
%each Newton step. We implemented this by finite elements or by a combination of finite elements and a Runge-Kutta scheme respectively.

The modeling by stochastic differential equations became standard in 
financial econometrics since the work of Black\& Scholes \cite{BS:73}. 
The parametric and non-parametric estimation of drift and diffusion in ergodic models has 
attracted a lot of interest since then. We just mention the text book by 
Kutoyants \cite{Kutoyants:04} and references therein. More recent works on nonparametric 
estimation of the drift are those by Hoffmann \cite{Hoffmann:99} using wavelets, 
Spokoiny \cite{Spokoiny:00} using kernel methods, Gobet, Hoffmann \& Rei{\ss} using 
wavelet estimation of an eigenvalue-eigenfunction pair of the transition operator, 
Comte, Genon-Catalot \& Rozenholc 
\cite{CGCR:07} using penalized least squares, 
Schmisser \cite{Schmisser:13} applying penalized least squares to high 
dimensional problems, Papaspiliopoulos et al.~\cite{PPRS:12}, Pokern, 
Stuart \& van Zanten \cite{PSZ:12} using Bayesian methods. 
%Many of these references are concerned with continuous time observations of one path 
%or discrete observations in the limit $\Delta_t\to 0$. 
A parametric estimator related to our approach was developed by Hurn, Jeismann 
\& Lindsay \cite{HJL:07}. They propose a maximum likelihood estimator 
which relies on the computation of \eqref{eq:fokker-planck_stationary} 
by finite elements. Due to a parametric model for $\bmu$ their 
problem is not ill-posed. Furthermore, we mention 
Cr{\'e}pey \cite{Crepey:031,Crepey:032}, Egger\& Engl 
\cite{EggerEngl:05} and De Cezaro, Scherzer \& Zubelli \cite{CSZ:12} 
for nonparametric volatility estimation using partial differential 
equations.

We will show convergence in expectation results with rates as 
$n\to\infty$ both for generalized Tikhonov regularization \eqref{eq:Tikh} 
and the iteratively regularized Newton method \eqref{eq:Gauss-Newton} 
by adapting corresponding results for inverse problems with Poisson 
data in \cite{HW:13,WH:12}. Here we make essential use of a version 
of Talagrand's concentration inequality due to Massart \cite{massart:00}. 

The iteratively regularized Gau{\ss}-Newton method with quadratic 
penalty and quadratic data fidelity term was suggested by 
Bakushinski{\u\i} \cite{bakush:92} and further analyzed by Blaschke, 
Neubauer \& Scherzer \cite{BNS:97} and Hohage \cite{hohage:97} for low 
order H\"older or logarithmic source conditions, respectively. Further 
references can be found in the monographs of Bakushinski{\u\i} \&
Kokurin \cite{BK:04b} and Kaltenbacher, Neubauer \& Scherzer 
\cite{KNS:08}. Regularization with general convex penalty terms 
have been recently investigated in a number of papers. 
We just mention Eggermont \cite{eggermont:93}, Burger \& Osher 
\cite{BO:04}, Resmerita \cite{resmerita:05}, Hofmann et al.\ 
\cite{HKPS:08}, and Scherzer et al. \cite{scherzer_etal:09}.  
Regularization methods for linear ill-posed problems with general data 
fidelity term like the log likelihood functional $\calS$ or the Kullback-Leibler 
divergence have been studied by Resmerita, Anderssen \cite{RA:07} and by 
Benning, Burger \cite{BB:11}. Linear and 
nonlinear Tikhonov regularization with general data fidelity terms has 
been investigated by Flemming \cite{Flemming:10}. 

The remainder of this paper is organized as follows: 
In the next section we present some  
properties of the Fokker-Planck equation and prove a tangential cone 
condition for the corresponding forward operator $F$.  
In Section~\ref{sec:general_convergence} general convergence rates results 
for variational regularization methods with Kullback-Leibler-type 
data fidelity and convex penalty term are presented. 
These results are applied to our 
estimator of the drift in Section~\ref{sec:convergence_sde}. 
Results of numerical simulations are shown in Section~\ref{sec:numerics} 
before we end this paper with some conclusions.

\section{Fokker-Planck equation}\label{sec:Fokker-Planck}

In this section we collect some properties of the stationary Fokker-Planck equation and prove the $L^2$ tangential cone condition for the corresponding operator $F$. We consider this equation on a bounded Lipschitz domain 
$\manifold\subset\mathbb{R}^d$ with the no-flux boundary condition. I.e. in terms of probability densities no probability mass enters or leaves through the boundary. It is the natural boundary condition for the Fokker-Planck equation:
\begin{equation}\label{eq:fokker-planck_stationary}
\eqalign{
\Div\left(-\bmu \uF + \half   \sigma\sigma^{\top} \grad \uF\right)
=0\qquad \mbox{in }\manifold\\
- \uF  (\bmu\cdot \mathbf{n})
+\half\left(\sigma\sigma^{\top} \grad \uF\right) \cdot \mathbf{n}
= 0 \qquad \mbox{on} \; \partial\manifold\\
\int_{\manifold} \uF(x) dx = 1.
}
\end{equation}
We assume that $\bmu\in L^{\infty}(\manifold,\mathbb{R}^d)$ and 
$\sigma\in L^{\infty}(\manifold)^{d\times d}$ with well-defined 
$L^{\infty}$ traces on $\partial D$ which appear in the 
boundary condition. Moreover, we assume  
that there exists a constant $C_\sigma > 0$ such that 
\begin{eqnarray}
\label{eq:unifom_elliptic}
&&|\sigma(x)^{\top}\xi|_2 \ge C_\sigma |\xi|_2 \quad \mbox{for all } \xi \in \R^d, \mbox{ and all } x \in \overline{\manifold}. 
%&&\|\sigma\sigma^{\top}\|_\infty := \sup_{x,t} \|\sigma(t,x)\sigma(t,x)^{\top}\| = \sup_{t,x, |\xi|_2 \le 1} 
%|[\sigma(t,x)\sigma(t,x)^{\top}] \xi|_{2}<\infty.
\end{eqnarray}
Let us comment on the natural boundary condition of the Fokker-Planck equation:
\begin{itemize}
\item In case $d=1$ we can assume w.l.o.g.\ that 
$\manifold= (-1,1)$. Extend $\bmu$ by $\bmu(x):= \bmu(1)$ and 
$\bmu(-x):=\bmu(-1)$ for $x>1$ and similarly for $\sigma$. 
Since the constant coefficient differential equation 
$-\bmu \uF'+\frac{\sigma^2}{2}\uF''=0$ with $\bmu\neq 0$ has the linearly independent solutions $1$ and $\exp\paren{\frac{2\bmu}{\sigma^2}x}$, 
the Fokker-Planck equation on $\mathbb{R}$ has an integrable solution 
if and only if $\bmu(1)<0$ and $\bmu(-1)>0$. In this case every integrable 
solution satisfies
\[\fl
\uF(x) =\uF(1)\exp\paren{\frac{2\bmu(1)}{\sigma(1)^2}(x-1)},\quad 
\uF(-x) =\uF(-1)\exp\paren{\frac{2\bmu(-1)}{\sigma(-1)^2}(1-x)},\quad 
x\geq 1.
\]
Therefore, these solutions satisfy the boundary condition in \eqref{eq:fokker-planck_stationary}. Hence, the restrictions of solutions to 
\eqref{eq:fokker-planck_stationary_unbounded} restricted to $\manifold=(-1,1)$ 
are solutions to \eqref{eq:fokker-planck_stationary} up to a scaling factor, 
i.e.\ the boundary condition is an exact transparent boundary condition. 
This is how the boundary condition will be interpreted in our numerical 
experiments. 
\item For $d>1$ 
%there are often no canonical extensions of $\bmu$ to the exterior domain, and 
exact transparent boundary conditions 
are always non-local. Since the boundary condition in 
\eqref{eq:fokker-planck_stationary} is local, we may at best hope for 
convergence to a solution of the Fokker-Planck equation in $\mathbb{R}^d$ 
as the size of $\manifold$ tends to $\infty$. 
\item In other applications, e.g.\ diffusion in biological cells 
the solution paths $\bfX_t$ are naturally contained in a subdomain $\manifold$ of 
$\mathbb{R}^d$. In this case the behavior at the boundary has 
to be modeled separately. E.g.\ when a path hits the boundary, it 
may be reflected in a certain way with a certain probability and otherwise 
destroyed. As discussed in \cite{SSOH:08,Schuss:10} and references 
therein, the behavior of the probability densities at the boundary may
be rather complex involving boundary layers, but no-flux boundary 
conditions often appear as limiting model. 
\end{itemize}

The weak formulation of the elliptic problem \eqref{eq:fokker-planck_stationary} is to find $\uF \in H^1(\manifold)$
such that
\begin{equation}\label{eq:stationary_weak}
\int_\manifold \uF dx=1,\qquad a_{\bmu}(\uF,\vTest) =0 
\quad \mbox{for all } \vTest \in H^1(\manifold) 
\end{equation}
where
\[
a_{\bmu}(\uF,\vTest):=\int_\manifold \paren{-\bmu \uF \cdot \grad\vTest + \half \sigma\sigma^{\top} \grad\uF \cdot \grad\vTest }dx.
\]
Let $L_{\bmu}: H^1(\manifold) \to H^{-1}_0(\manifold)$ denote the operator associated to $a_{\bmu}$, i.e.\ 
$\langle L_{\bmu}\uF,\vTest\rangle = a_{\bmu}(\uF,\vTest)$ for all $\uF,\vTest\in H^1(\manifold)$. 
It was proven by Droniou and V{\'a}zquez \cite{DV:10} that every function in the kernel of $L_{\bmu}$ is either a.e.\ positive, a.e.\ negative, or a.e.\ $0$. Therefore, the kernel is either trivial or one-dimensional. 
%Hence, the condition $\int F(\bmu) dx = 1$ characterizes a unique solution if the kernel is not trivial. 
For the convenience of the reader we collect some further properties of $L_{\bmu}$ all of which are more or less explicitly 
contained in \cite{DV:10}. 
%The following G\aa rding inequality shows that $L_{\bmu}$ is a Fredholm operator. 

\begin{lemma}
Assume \eqref{eq:unifom_elliptic} for $\sigma$ and let $\bmu \in L^\infty(\manifold,\mathbb{R}^d)$. 
\begin{enumerate}
\item
The following G\aa rding inequality holds with $\gamma > \|\bmu\|_\infty^2/(2C_\sigma)$ and 
$ 0 < c < \min\left\{\gamma - \frac{\|\bmu\|_\infty^2}{2C_\sigma}, \frac{C_\sigma}{2}-\frac{\|\bmu\|_\infty^2}{4\gamma }\right\}$
\[
a_{\bmu}(\uF, \uF) + \gamma \|\uF\|_{L^2}^2 \geq c \|\uF\|_{H^1}^2,\qquad \uF\in H^1(\manifold). 
\]
\item Eq.~\eqref{eq:stationary_weak} has a unique solution. 
\item Let %$H^{-1}_{0,\diamond}(\manifold) := \{g \in H^{-1}_0(\manifold) | \langle g, 1 \rangle = 0 \}$ and  
$H^1_{\diamond}(\manifold) := \{\uF \in H^1(\manifold) | \int \uF dx = 0 \}$, let 
$\tilde{a}_{\bmu}:H^1_{\diamond}(\manifold)\times H^1_{\diamond}(\manifold)\to \R$ denote the restriction of $a_{\bmu}$ to $H^1_{\diamond}(\manifold)$, 
and let $\tilde{L}_{\bmu}:H^1_{\diamond}(\manifold)\to H^1_{\diamond}(\manifold)^*$ denote the operator associated to $\tilde{a}_{\bmu}$. 
Then $\tilde{L}_{\bmu}$ 
%\[
%\tilde{L}_{\bmu}:H^1_{\diamond}(\manifold)\to H^{-1}_{0,\diamond}(\manifold),\qquad \tilde{L}_{\bmu}\uF := L_{\bmu}\uF
%\]
is bijective and has a bounded inverse. %which is denote (with a slight abuse of notation) by $L_{\bmu}^{-1}$. 
%For every $g\in H^{-1}_{0,\diamond}(\manifold)$, 
%$L_{\bmu}^{-1}g$ is uniquely characterized by 
%\begin{eqnarray*}
%&&-\Div\paren{- \bmu \uF + \half \sigma\sigma^{\top} \, \grad \uF}=g\qquad \mbox{in }\manifold\\
%&&\sigma\sigma^{\top} \grad \uF \cdot \mathbf{n} = 2 \bmu \uF \cdot \mathbf{n} \qquad \mbox{on } \partial\manifold\\
%&&\int \uF dx = 0.
%\end{eqnarray*}
\end{enumerate}
\end{lemma}

\begin{proof}
1) We have
\begin{eqnarray*}
a_{\bmu}(\uF, \uF) +\gamma\|\uF\|_{L^2}^2 
%&= \int_\manifold \Div\paren{\bmu \uF - \half \sigma\sigma^{\top} \grad \uF}\uF dx +\gamma\|\uF\|_{L^2}^2\\
&= \int_\manifold - \bmu \uF \grad \uF + \half  \left|\sigma^{\top}\grad \uF\right|_2^2 dx +\gamma\|\uF\|_{L^2}^2\\
&\geq - \|\bmu\|_\infty \|\uF\|_{L^2} \|\grad \uF\|_{L^2}  + \frac{C_\sigma}{2} \|\grad \uF\|_{L^2}^2 +\gamma\|\uF\|_{L^2}^2\\
&\geq \left(\gamma -\frac{\|\bmu\|_\infty^2}{4\varepsilon}\right)\|\uF\|_{L^2}^2 + \left(\frac{C_\sigma}{2}- \varepsilon\right) \|\grad \uF\|_{L^2}^2.
\end{eqnarray*}
The last step uses Young's inequality $ab\leq a^2/(4\epsilon)+\epsilon b^2$, which holds for $a,b\geq 0$ and $\varepsilon > 0$. 
Choosing $\varepsilon < C_\sigma/2$ and $\gamma > \|\bmu\|_\infty^2/(4\varepsilon)$ gives the G\aa rding inequality.\\
2) As a consequence of part 1, $L_{\bmu}$ is a Fredholm operator of index $0$, i.e.\ 
$\dim(\ker(L_{\bmu}))=\dim(\ran(L_{\bmu})^{\perp})$ (where orthogonality is understood with respect to the dual pairing 
of $H^1(\manifold)$ and $H^{-1}_0(\manifold)$) and $\ran(L_{\bmu})$ is closed.  
As argued above, $\dim(\ker(L_{\bmu}))\in\{0,1\}$. As $a_{\bmu}(\uF,1)=0$ for all $\uF\in H^1(\manifold)$, i.e.
$1\in \ran(L_{\bmu})^{\perp}$, we have $\dim(\ker(L_{\bmu}))=1$. Since the elements of $\ker(L_{\bmu})$ are 
positive a.e.\ or negative a.e., there exists a unique $\uF\in\ker(L_{\bmu})$ satisfying $\int_{\manifold}\uF\,dx=1$. \\
3) We also have $\dim(\ran(L_{\bmu})^{\perp})=1$, so by the proof of part 2 $\ran(L_{\bmu})=\{1\}^{\perp} = H^{1}_{\diamond}(\manifold)^*$ as 
$\ran(L_{\bmu})$ is closed. By the characterization of $\ker(L_{\bmu})$, the operator $L_{\bmu}$ is injective on 
$H^1_{\diamond}(\manifold)$. Moreover, $\ran(\tilde{L}_{\bmu})=\ran(L_{\bmu})$ as 
$H^1_{\diamond}(\manifold)\oplus\mathrm{span}\{1\}=H^1(\manifold)$, so $\tilde{L}_{\bmu}$ is surjective. 
Boundedness of $\tilde{L}_{\bmu}^{-1}$ follows from the open mapping theorem. %For explicit bounds we refer to \cite{DV:10}.
\end{proof}

The differentiability of $F$ and the tangential cone condition stated in the next theorem
are crucial for the Gau{\ss}-Newton method.

\begin{theorem}\label{the:elliptic}
The operator $F : L^\infty(\manifold,\mathbb{R}^d) \to L^2(\manifold)$ is 
Fr\'echet differentiable, and 
$F'[\bmu]\bh=\uF_{\bmu,\bh}'$ where $\uF_{\bmu,\bh}'\in H^1_{\diamond}(\manifold)$ is the unique solution to 
the variational problem
\begin{equation}\label{eq:stationary_der}
\tilde{a}_{\bmu}(\uF'_{\bmu,\bh},\vTest) =  \int_{\manifold} F(\bmu) \;\bh \cdot \grad\vTest\,dx,\qquad \vTest\in H^1_{\diamond}(\manifold).
\end{equation}
% Equivalently, $\uF'_{\bmu,\bh}$ is the solution to the boundary value problem
% \begin{eqnarray*}
% && \Div\paren{- \bmu \uF'_{\bmu,\bh} + \half \sigma\sigma^{\top} \, \grad \uF'_{\bmu,\bh}} = \Div(\bh F(\bmu)) \\
% &&\frac{1}{2}\sigma\sigma^{\top} \grad \uF'_{\bmu,\bh} \cdot \mathbf{n} - \bmu \uF'_{\bmu,\bh} = F(\bmu) \bh \cdot \mathbf{n} \qquad \mbox{on } \partial\manifold\\
% &&\int \uF'_{\bmu,\bh} dx = 0.
% \end{eqnarray*}
Furthermore, the strong tangential cone condition holds true:
\begin{equation}\label{eq:tang_cone_Fokker}
\|F(\bmu+\bh) - F(\bmu) - F'[\bmu]\bh\|_{L^2} \leq \tilde{C}_{\bmu} \|\bh\|_\infty \|F(\bmu + \bh) - F(\bmu)\|_{L^2}
\end{equation}
for all $\bmu,\bh\in L^{\infty}(\manifold,\mathbb{R}^d)$ with 
$\tilde C_{\bmu}:=\|\tilde{L}_{\bmu}^{-1}\|$.
\end{theorem}

\begin{proof}
Note that $\tilde \uF := F(\bmu + \bh) - F(\bmu)$ belongs to $H^1_{\diamond}(\manifold)$ and satisfies 
\begin{eqnarray*}
\tilde{a}_{\bmu}(\tilde \uF, \vTest ) = \int_{\manifold} (F(\bmu)+\tilde \uF)\;\bh \cdot \grad\vTest\,dx,\qquad \vTest\in H^1_{\diamond}(\manifold).
%-\langle \Div\big(\bh (F(\bmu)+\tilde \uF), \vTest \rangle + \langle \Tr_\n(\bh(F(\bmu)+\tilde \uF)), \Tr \vTest \rangle
\end{eqnarray*}
%for all $\vTest \in H^1_{\diamond}(\manifold)$. 
% Using the adjoint $\Tr^* : H^{-1/2}(\partial\manifold) \to  H^{-1}_0(\manifold)$ of the trace operator 
% $\Tr : H^1(\manifold) \to H^{1/2}(\partial\manifold)$ we can write
% \begin{equation}\label{eq:aux_trace2}
% \tilde \uF = \tilde{L}_{\bmu}^{-1}\left(-\Div\big(\bh (F(\bmu)+\tilde \uF)\big) + \Tr^*\Tr_\n( \bh(F(\bmu)+\tilde \uF))\right).
% \end{equation}
For $\vTest\neq 0$ the functional on the right hand side is bounded by
\[
\fl
\frac{1}{\|\vTest\|_{H^1}}\int_{\manifold} (F(\bmu)+\tilde \uF)\;\bh \cdot \grad\vTest\,dx
%\leq \|\bh (F(\bmu)+\tilde \uF)\|_{L^2}
\le  \|\bh\|_\infty \|F(\bmu)+\tilde \uF\|_{L^2}
\le  \|\bh\|_\infty \left(\|F(\bmu)\|_{L^2} + \|\tilde \uF\|_{H^1}\right)
\]
Therefore, 
\(
\|\tilde \uF\|_{H^1} %&= \left\|\tilde{L}_{\bmu}^{-1}\left(\Div\big(\bh (F(\bmu)+\tilde \uF)\big) - \tilde{L}_{\bmu}^{-1} \Tr^*\Tr_\n( \bh(F(\bmu)+\tilde \uF))\right)\right\|_{H^1}\\
% &\le \tilde C_{\bmu} \|\Div\big(\bh (F(\bmu)+\tilde \uF)\big) - \Tr^*\Tr_\n(\bh(F(\bmu)+\tilde \uF))\|_{H^{-1}_0}\\
% & = \tilde C_{\bmu} \sup_{\|\vTest\|_{H^1}\le 1} \langle \Div\big(\bh (F(\bmu)+\tilde \uF)\big), \vTest \rangle_{L^2} 
% - \langle \Tr_\n(\bh(F(\bmu)+\tilde \uF)), \Tr\vTest \rangle_{L^2}\\
% & = \tilde C_{\bmu} \sup_{\|\vTest\|_{H^1}\le 1} \langle \bh (F(\bmu)+\tilde \uF), \grad\vTest \rangle_{L^2}\\
% &\le \tilde C_{\bmu} \|\bh (F(\bmu)+\tilde \uF)\|_{L^2}
% \le \tilde C_{\bmu} \|\bh\|_\infty \|F(\bmu)+\tilde \uF\|_{L^2}\\
\le \tilde C_{\bmu} \|\bh\|_\infty \left(\|F(\bmu)\|_{L^2} + \|\tilde \uF\|_{H^1}\right),
\)
%with $\tilde C = \tilde C_{\bmu} (1+C_{\Tr^*}C_{\Tr_\n})$. 
which implies
\[
(1 - \tilde C_{\bmu}\|\bh\|_\infty) \|\tilde \uF\|_{H^1} \le \tilde C_{\bmu} \|\bh\|_\infty \|F(\bmu)\|_{L^2}.
\]
% \[
% (1 - \tilde C\|\bh\|_{W^{1,\infty}(\Div;\manifold)})\|\tilde \uF\|_{H^1(\manifold)} \le \tilde C \|\bh\|_{W^{1,\infty}(\Div;\manifold)} \|F(\bmu)\|_{H^1}
% \]
Hence, $F$ is continuous since $\|F(\bmu+\bh) -F(\bmu)\|_{H^1} = \|\tilde \uF\|_{H^1}$ tends to $0$ as 
$\|\bh\|_\infty$ tends to $0$.
%Now we can show the tangential cone condition. 
%Note that \eqref{eq:aux_trace1} holds true with $\tilde \uF$ replaced by the derivative $\uF'_{\bmu,\bh}$. 
%Like $\tilde \uF$ the derivative $\uF'_{\bmu,\bh}$ fails to fulfill the boundary condition for $L_{\bmu}$ by the term $\bh F(\bmu)\cdot \n$. 
%Therefore,  in analogy to \eqref{eq:aux_trace2} we obtain from \eqref{eq:stationary_der} the identity
As 
\[
\tilde{a}_{\bmu}(\tilde \uF-\uF_{\bmu,\bh}, \vTest ) 
= \int_{\manifold} \tilde \uF\;\bh \cdot \grad\vTest\,dx,\qquad \vTest\in H^1_{\diamond}(\manifold),
\]
a similar estimate of the right hand side as above yields the bound
\begin{eqnarray*}
\fl
\|F(\bmu + \bh) - F(\bmu) - \uF'_{\bmu,\bh}\|_{L^2} &= \|\tilde \uF - \uF'_{\bmu,\bh}\|_{L^2}
\leq \|\tilde \uF - \uF'_{\bmu,\bh}\|_{H^1}
% &= \|\tilde{L}_{\bmu}^{-1}(-\Div(\bh\tilde \uF) + \Tr^*\Tr_\n(\bh \tilde \uF))\|_{L^2}\\
% & \leq \|\tilde{L}_{\bmu}^{-1}(-\Div(\bh\tilde \uF) + \Tr^*\Tr_\n(\bh \tilde \uF))\|_{H^1}\\
% & \leq \tilde C_{\bmu} \|-\Div(\bh\tilde \uF) + \Tr^*\Tr_\n(\bh \tilde \uF)\|_{H^{-1}_0}\\
% & = \tilde C_{\bmu} \sup_{\|\vTest\|_{H^1}\le 1} \langle -\Div\big(\bh \tilde \uF\big), \vTest \rangle_{L^2} 
% + \langle \Tr_\n(\bh\tilde \uF), \Tr\vTest \rangle_{L^2}\\
% & = \tilde C_{\bmu} \sup_{\|\vTest\|_{H^1}\le 1} \langle \bh \tilde \uF, \grad \vTest \rangle_{L^2}
% \le \tilde C_{\bmu} \|\bh \tilde \uF\|_{L^2}
\le \tilde C_{\bmu} \|\bh\|_\infty \|\tilde \uF\|_{L^2},
% & \leq \tilde C_{\bmu} \left(\|\tilde \uF \Div \bh\|_{H^{-1}} + \|\bh \cdot \grad\tilde \uF\|_{H^{-1}} + C_{\Tr^*}C_{\Tr_\n} \|\bh \tilde \uF\|_{H(\Div;\manifold)} \right)\\
% & \leq \tilde C_{\bmu} \left(\|\tilde \uF\|_{L^2} \|\Div \bh\|_\infty + \|\bh\|_\infty \|\tilde \uF\|_{L^2} \right)\\
% & \leq \tilde C_{\bmu} \|\bh\|_{W^{1,\infty}(\Div;\manifold)} \left( \|\tilde \uF\|_{L^2} + C_{\Tr^*}C_{\Tr_\n} \|\tilde \uF\|_{H^1} \right)\\
% & \leq  \tilde C_{tcc} \|\bh\|_{W^{1,\infty}(\Div;\manifold)} \|\tilde \uF\|_{L^2}
\end{eqnarray*}
which shows the tangential cone condition. Together with the continuity of $F$ this implies that 
$F$ is Fr\'echet differentiable, and $F'[\bmu]\bh=\uF'_{\bmu,\bh}$. 
%with some constant $\tilde C_{tcc}$ for $\|\bh\|_{W^{1,\infty}(\Div;\manifold)}$ small enough. 
% This inequality together with the continuity of $F$ show the assertion. The last step in the estimation is implied by the regularity estimate
% \[
% \left(1 + 2C_\sigma^{-2}(\|\bmu\|_\infty + \|\bh\|_\infty)\right) \|\tilde \uF\|_{L^2} \ge \|\tilde \uF\|_{H^1}
% \]
% This is a consequence of the weak formulation which implies that for every $\vTest \in H^1$
% \begin{eqnarray*}
% \left|\left\langle (\bmu+\bh)F(\bmu+\bh) - \bmu F(\bmu), \grad \vTest \right\rangle_{L^2}\right| &= \left| \left\langle \half \sigma\sigma^{\top} \grad \tilde \uF, \grad \vTest \right\rangle_{L^2} \right|\\
% &\ge \left|\frac{C_\sigma^2}{2} \left\langle \grad \tilde \uF, \grad \vTest \right\rangle_{L^2}\right|.
% \end{eqnarray*}
% Choosing $\vTest = \tilde \uF$ and applying Cauchy-Schwarz yields the regularity estimate.
\end{proof}

\begin{example}\label{ex:explicit}
If $\bmu$ has a representation of the form 
\begin{equation}\label{eq:potential_repr}
\bmu = \sigma\sigma^\top\grad\phi
\end{equation}
for some potential $\phi$ the solution of the stationary Fokker-Planck equation \eqref{eq:stationary_weak} is given explicitly by 
\[
\uF = \frac{1}{\int_\manifold \exp(2 \phi) \, dx}\exp(2 \phi),
\]
since
\begin{eqnarray*}
\grad \uF &= \frac{2}{\int_\manifold \exp(2 \phi) \, dx} \grad\phi \exp\left(2 \phi\right)
=  2(\sigma\sigma^{\top})^{-1} \bmu \uF.
\end{eqnarray*}
The normalization constant $\int_\manifold \exp(2 \phi) \, dx$ ensures that 
$\uF$ is a density. In particular, we obtain the following explicit
formula for the inverse of $F$:
\begin{equation}\label{eq:explicit_inverse}
\bmu = \frac{\sigma\sigma^{\top}\grad u}{2u}.
\end{equation} 
%which has been used for the reconstruction of $\bmu$ in \cite{}.
The methods discussed below do not rely on this formula and the assumption 
\eqref{eq:potential_repr}. 
\end{example}

\section{General convergence results for inverse problems 
with i.i.d.\ sample data}
\label{sec:general_convergence}

In this section we consider the following general setting:
\begin{itemize}
\item $\Xspace$ is a Banach space, $\frakB \subset \Xspace$ a convex subset, $\manifold\subset\mathbb{R}^d$ a bounded 
Lipschitz domain, and $H^s(\manifold)$ with $s>\frac{d}{2}$ an $L^2$-based  Sobolev space. 
\item The range of operator $\Opg:\frakB\to H^s(\manifold)$ consists of probability densities, i.e.\ $\Opg(\sol)\geq 0$ and 
$\int_{\manifold} \Opg(\sol)\,dx=1$ for all $\sol\in\frakB$. 
\item There exists $R>1$ such that $\sup_{\sol\in\frakB}\|\Opg(\sol)\|_{H^s}\leq R$.
\item $\sol^{\dag}\in\frakB$ is the exact solution, $\uF^{\dag}:=\Opg(\sol^{\dag})$, and 
observations are described by independent random variables $\bfY_1,\dots,\bfY_n$ with 
density $\uF^{\dag}$. Recall the definition of the empirical measure 
$\Phi_n$ in \eqref{eq:empirical_meas}.
\end{itemize}

\subsubsection*{A concentration inequality.}
Note that 
\begin{eqnarray*}
&& \EW{\int_{\manifold}\varphi\,d\Phi_n} = \int_{\manifold}\varphi\uF^{\dag}\,dx,\qquad\mbox{and}\qquad 
\Var{\int_{\manifold}\varphi\,d\Phi_n} 
= \frac{1}{n} \int_{\manifold}\varphi^2\uF^{\dag}\,dx
\end{eqnarray*}
whenever the right hand sides are well-defined. We will need a concentration inequality which is uniform in $\varphi$. 
Our starting point is a version of the concentration inequality in the seminal work by 
Talagrand \cite{Talagrand:96}, which is due to Massart \cite{massart:00}
and has explicit constants. 
In our notation a special case of this inequality can be stated 
as follows:
\begin{theorem}[Theorem 3 in \cite{massart:00}]\label{theo:massart}
Let $\mathcal{F}\subset L^{\infty}(D)$ be a countable family of functions with $\|\varphi\|_{\infty} \leq b$ 
for all $\varphi\in\mathcal{F}$. Moreover, let 
\[
Z:=n\sup_{\varphi\in\mathcal{F}}\left|\int_{\manifold}\varphi(d\Phi_n - \uF^{\dag}dx)\right|
\]
and $v:=n\sup_{\varphi\in\mathcal{F}}\int_{\manifold}\varphi^2\uF^{\dag}\,dx$. Then 
\[
\Prob{Z\geq (1+\epsilon)\EW{Z}+\sqrt{8v\xi}+\kappa(\epsilon)b\xi} \leq \exp(-\xi)
\] 
for all $\epsilon,\xi>0$ where $\kappa(\epsilon)= 2.5+32/\epsilon$. 
\end{theorem} 
Massart also proved a similar inequality for the left tail of $Z$, 
but we only need the inequality above, so we might rather speak of 
a deviation inequality. 

In analogy to \cite{WH:12} where similar results were derived  
using a concentration inequality for Poisson processes 
in \cite{RB:03} instead of Theorem \ref{theo:massart}, 
we show the following corollary:
\begin{corollary}\label{cor:conc}
There exists a constant $C_{\rm c} \ge 1$ depending only on $\manifold$ and 
$s$ such that for $\rho \ge R C_{\rm c}$ and for all $n \in \N$
\begin{equation}\label{eq:conc_ieq}
\Prob{\sup_{\|\varphi\|_{H^s(\manifold)}\leq R} \left| \int_\manifold \varphi \left(d\Phi_n - \uF^\dag dx \right) \right| 
\ge \frac{\rho}{\sqrt{n}}} \le \exp\left(-\frac{\rho}{R C_{\rm c}}\right).
\end{equation}
\end{corollary}

\begin{proof} \emph{(Sketch)}
The most difficult part in the derivation of Corollary \ref{cor:conc} 
from Theorem \ref{theo:massart} is the estimation of $\EW{Z}$. 
In analogy to \cite[Lemma A.2]{WH:12} we can prove that
\[
\EW{Z}\leq \sqrt{n}C_1R
\]
with a constant $C_1$ depending only on $s$ and $\manifold$. 
As $H^s(\manifold)$ is continuously embedded in $L^{\infty}(\manifold)$,  
we have $\|\varphi\|_{\infty}\leq C_2R$ for all 
$\varphi\in H^s(\manifold)$ with $\|\varphi\|_{H^s}\leq R$
where $C_2$ is the norm of the embedding operator. 
Moreover, $v\leq  n (C_2R)^2$ as $\|\uF^\dag\|_{L^1}=1$. 
Using the separability of balls in $H^s(\manifold)$ and choosing
$\epsilon=1$ in Theorem~\ref{theo:massart} we obtain
\[\fl
\Prob{\sup_{\|\varphi\|_{H^s(\manifold)}\leq R} \left| \int_\manifold \varphi \left(d\Phi_n - \uF^\dag dx \right) \right| 
\ge \paren{\frac{2C_1}{\sqrt{n}}+\frac{C_2\sqrt{8\xi}}{\sqrt{n}}
+ \frac{34.5C_2\xi}{n}}R}\leq \exp(-\xi).
\]
As $\frac{1}{n}\leq \frac{1}{\sqrt{n}}$ and $\sqrt{\xi}\leq \xi$ for
$\xi\geq 1$, this yields \eqref{eq:conc_ieq} with 
$C_{\rm c}:= 2C_1+(34.5+\sqrt{8})C_2$ and 
$\rho = RC_{\rm c}\xi$. 
\end{proof}

\subsubsection*{Distance measures.}
To state our convergence theorems we need both distance measures in 
$\Xspace$ and $L^1(\manifold)$. As usual for variational regularization methods convergence rates are given with respect to the Bregman distance associated to 
the penalty term as loss function. The Bregman distance with respect to $\calR$ 
and $\sol^* \in \partial\calR(\sol^\dag)$ is
\[
D_\calR^{\sol^*}(\sol,\, \sol^\dag) := \calR(\sol) - \calR(\sol^\dag) - \langle \sol^*, \sol - \sol^\dag \rangle.
\]
Recall that for quadratic penalty in Hilbert spaces we have 
$D_\calR^{\sol^*}(\sol,\, \sol^\dag)= \|\sol-\sol^{\dag}\|^2$. 
In general, $D_{\calR}^{\sol^*}$ is nonnegative with 
$D_{\calR}^{\sol^*}(\sol^{\dag},\sol^{\dag})=0$, but it is neither 
symmetric nor does it satisfy a triangle inequality. 

The distance measure in $L^1(\manifold)$ which corresponds to the negative log-likelihood introduced in \eqref{eq:S0} 
is the Kullback-Leibler divergence
\[
\KL(\uF;\vTest) := \int_\manifold \vTest-\uF  
-\uF \ln\left(\frac{\vTest}{\uF}\right) \, dx
\]
with the convention $0\ln 0:=0$ and $\ln(x):=-\infty$ for $x\leq 0$. 
Note that $\KL(\uF^{\dagger};\vTest)= 
\EW{\calS_0(\Phi_n;\vTest)-\calS_0(\Phi_n;\uF^{\dagger})}$, in other 
words $\KL$ is the expectation of the negative log-likelihood functional with an additive constant chosen 
in a way such that $\KL(\uF^{\dagger};\vTest)\geq 0$ 
for all $\vTest$ and $\KL(\uF^{\dagger};\uF^{\dagger})=0$.  
If $\uF$ and $\vTest$ are probability densities, the formula above 
simplifies to $\KL(\uF;\vTest) = \int_\manifold \uF \ln\left(\vTest/\uF\right) dx$, 
but since the values of the linearization of $F$ are not densities 
in general, we have to use the general formula. 

Note that 
\[
\calS_0(\Phi_n;\vTest)-\calS_0(\Phi_n;\uF^{\dag})-\KL(\uF^{\dagger};\vTest)
= \int -\ln \frac{v}{\uF^{\dagger}} \left(d\Phi_n-\uF^{\dagger}dx\right).
\]
To prove rates of convergence we %reduce the problem to a deterministic inverse problem by bounding 
have to bound the absolute value of the right hand side with sufficiently large probability. 
In principle, this can be done by applying 
Corollary \ref{cor:conc} with $\varphi = -\ln \frac{v}{\uF^{\dagger}}$. 
However, this corollary is only applicable if we have uniform bounds 
$0<c\leq \frac{v}{\uF^{\dagger}}\leq C<\infty$ for all $v\in F(\mathfrak{B})$, 
which is not always the case.
Therefore, we introduce a shift parameter $\shift>0$ and use 
$\KL(\uF^{\dagger}+\shift,\vTest+\shift)$ as 
limiting data fidelity term and the corresponding empirical data
fidelity term 
\[
\calS_{\shift}(\Phi_n;\vTest) = \int_{\manifold} \vTest dx 
- \int_{\manifold} \ln (\vTest+\shift)(d\Phi_n+\shift dx)
\] 
such that 
\[
\fl
\calS_{\shift}(\Phi_n;\vTest)-\calS_{\shift}(\Phi_n;\uF^{\dag})
-\KL(\uF^{\dagger}+\shift;\vTest+\shift)
= \int -\ln \frac{v+\shift}{\uF^{\dagger}+\shift} \left(d\Phi_n-\uF^{\dagger}dx\right).
\]
Now we can bound 
\[
\mathrm{err}:=\sup_{v\in \Opg(\mathfrak{B})}\left|\calS_{\shift}(\Phi_n;\vTest)-\calS_{\shift}(\Phi_n;\uF^{\dag})
-\KL(\uF^{\dagger}+\shift;\vTest+\shift)\right|
\]
with high probability using Corollary \ref{cor:conc} since 
$\sup_{\vTest\in \Opg(\mathfrak{B})}\|-\ln \frac{\vTest+\shift}{\uF^{\dagger}+\shift}\|_{H^s}<\infty$ 
under our assumptions.

\subsubsection*{Convergence rate results.}
To obtain rates of convergence we need some kind of smoothness condition 
on the solution. Source conditions are commonly used for this purpose.
In the regularization theory for Banach spaces they are formulated as variational 
inequalities (see \cite{HKPS:08} 
and \cite{flemming:12b} for relations to other formulations of source conditions). 
We assume that there exists of a constant $\beta > 0$, 
$\sol^*\in\partial\calR(\sol^{\dagger})$ and a 
concave, strictly increasing function 
$\Lambda : [0 , \infty[ \; \to [0 , \infty[$ with $\Lambda(0) = 0$ 
such that
\begin{equation}\label{eq:source_condition}
\fl
\beta D_{\calR}^{\sol^*}(\sol,\, \sol^\dag) \le \calR(\sol) - \calR(\sol^\dag) + \Lambda\Big(\KL\paren{\uF^\dag+\shift; \Opg(\sol)+\shift}\Big) \quad \mbox{for all } \sol \in \frakB.
\end{equation}
The proof of the following theorem is now completely analogous 
to the proof of \cite[Theorem 4.3]{WH:12}, but we point out that in 
\cite[eq.~(10)]{WH:12} on the left hand side $\EW{\mathcal{S}(G_t;g^{\dagger})}$ 
should be replaced by $\mathcal{S}(G_t;g^{\dagger})$ and on the right hand side 
$\ln(g+\sigma)$ by $\ln \frac{g+\sigma}{g^{\dagger}+\sigma}$. 
\begin{theorem}\label{theo:tikh}
If $\uF^{\dagger}$ satisfies the variational source condition 
\eqref{eq:source_condition} for some $\shift>0$, the nonlinear Tikhonov regularization
\eqref{eq:Tikh} with $\calS = \calS_{\shift}$ has a global minimizer $\solhat_\alpha$, 
and the regularization parameter is chosen such that 
\begin{equation}\label{eq:apriori_alphaTikh}
\alpha^{-1} \in -\partial(- \Lambda)\left(\frac{2\rho}{\sqrt{n}}\right),
\end{equation}
then we have
\begin{equation}\label{eq:apriori_rates_tikh}
\EW{D_\calR^{\sol^*}(\solhat_{\alpha}, \sol^\dag)} 
= \mathcal{O}\left(\Lambda\left(\frac{1}{\sqrt{n}}\right)\right),
\qquad n\to\infty. 
\end{equation}
\end{theorem}
To prove convergence of the Newton-type iteration we additionally 
have to impose a tangential cone condition adapted to our data 
fidelity term. Let 
\[
\calT_{\shift}(\uF;v):=\left\{\begin{array}{ll}
\KL(\uF+\shift,v+\shift)&\mbox{ if } v\geq -\shift/2 \\
\infty&\mbox{ else.}
\end{array}\right.
\] 
We assume that for all $\sol, \altsol \in \frakB$
\begin{equation}\label{eq:KL-tcc}
\eqalign{
\frac{1}{C_{\rm tcc}} \calT_{\shift}\left(\uF^\dag; \Opg(\altsol)\right) 
- \eta \calT_{\shift} \left(\uF^\dag; \Opg(\sol)\right) 
&\leq \calT_{\shift} \left(\uF^\dag; \Opg(\sol) + \Opg'[\sol](\altsol-\sol)\right)\\
&\leq C_{\rm tcc} \calT_{\shift} \left(\uF^\dag; \Opg(\altsol)\right) 
+ \eta \calT_{\shift} \left(\uF^\dag; \Opg(\sol)\right)}
\end{equation}
with $\eta$ sufficiently small and $C_{tcc}>1$. 
We also set $\calS_{\shift}(\Phi_n;v):=\infty$ if $v\geq -\shift/2$. 
Then we can show in analogy to  \cite{HW:13}:
\begin{theorem}\label{theo:convergencerate}
Let assumptions \eqref{eq:source_condition}, \eqref{eq:KL-tcc} hold true. 
If $\solhatk$ is defined by the iteratively regularized Newton method \eqref{eq:Gauss-Newton} where $k \in \N$ is the largest index such that
\begin{equation}\label{eq:apriori_alpha}
\alpha_k^{-1} \le \sup -\partial(- \Lambda)\left(\frac{2\rho}{\sqrt{n}}\right),
\end{equation}
then
\begin{equation}\label{eq:apriori_rates}
\EW{D_\calR^{\sol^*}(\muhatk, \sol^\dag)} = \mathcal{O}\left(\Lambda\left(\frac{1}{\sqrt{n}}\right)\right).
\end{equation}
\end{theorem}

\begin{remark}
\begin{enumerate}
\item Related results exist for the iteratively regularized Gau{\ss}-Newton method with $L^2$ data fidelity term. Instead of \eqref{eq:KL-tcc}, these theorems assume the $L^2$ tangential cone condition \eqref{eq:L2tcc}. Results like this were proven by Kaltenbacher and Hofmann \cite{KH:10}, Hohage and Werner \cite{HW:13}, or Dunker et al. \cite{DFHJM:13}. The convergence rates for quadratic data fidelity terms compare to the rates  in \eqref{eq:apriori_rates}.

\item The selection rule \eqref{eq:apriori_alpha} uses \textit{a priori} information about the index function $\Lambda$ which is usually not available in practice. It was shown in \cite{HW:13} that a data driven Lepski{\u\i} type parameter choice can be used instead. Only a logarithmic factor gets lost in the resulting convergence rate:
\[
\EW{D_\calR^{\sol^*}(\widehat\sol_{k_{\mathrm{Lepskii}}}, \sol^\dag)} 
=\mathcal{O}\left(\Lambda\left(\frac{\ln(n^{-1})}{\sqrt{n}}\right)\right).
\]
\end{enumerate}
\end{remark}

\section{Convergence of the drift estimator}\label{sec:convergence_sde}
In order to apply Theorems \ref{theo:tikh} and 
\ref{theo:convergencerate} to the drift estimation problem with 
Poisson data we have to discuss the assumptions 
\eqref{eq:source_condition} and \eqref{eq:KL-tcc}. 
For this purpose we need the following estimates for the Kullback-Leibler divergence:  
\begin{lemma}\label{lem:L2-KL}
The inequality 
\begin{equation}\label{eq:KL_lower_bound}
\|\varphi - \psi\|_{L^2}^2 \le \left(\frac{2}{3}\|\varphi\|_\infty + \frac{4}{3} \|\psi\|_\infty \right) \KL(\varphi;\psi).
\end{equation}
holds for all nonnegative functions $\varphi, \psi \in L^\infty(\manifold)$ with 
$\varphi - \psi \in  L^2(\manifold)$. 
If $\psi$ is bounded away from $0$ then 
\begin{equation}\label{eq:KL_upper_bound1}
\KL(\varphi;\psi) \leq \left\|\frac{1}{\psi}\right\|_\infty \|\varphi - \psi\|_{L^2}^2.
\end{equation}
\end{lemma}

\begin{proof}
%If $\varphi(x) = 0$ and $\psi(x) \neq 0$ on a non-null set, $\KL(\varphi;\psi) = \infty$ and the lower bound is valid. Let us assume in turn that $\varphi(x) > 0$ and $\psi(x) > 0$ a.e.\ in $\manifold$. The function 
%\[
%g(x) := \left(\frac{2x}{3} + \frac{4}{3} \right) (x \ln x - x + 1) -(x-1)^2 \qquad \mbox{on the interval } [0, \infty),
%\]
%is strictly convex and has its minimum at $g(1) = 0$. By setting $x = \varphi / \psi$ we get
%\begin{eqnarray*}
%(\varphi - \psi)^2 &\le \left(\frac{2}{3} \varphi + \frac{4}{3} \psi\right)\left( \psi - \varphi - \varphi \ln \left(\frac{\psi}{\varphi}\right) \right)\\
%&\le \left(\frac{2}{3} \|\varphi\|_\infty + \frac{4}{3} \|\psi\|_\infty\right)\left( \psi - \varphi - \varphi \ln \left(\frac{\psi}{\varphi}\right) \right).
%\end{eqnarray*}
The lower bound can be found e.g.\ in \cite{BL:91}. 
The upper bound follows from the simple estimation $x-1 \ge \ln x$ 
which entails $(x-1)^2 \ge x \ln x - x + 1$. Setting $x = \varphi / \psi$ 
we get
\[
\frac{1}{\psi}(\varphi - \psi)^2 \ge \psi - \varphi - \varphi \ln \left(\frac{\psi}{\varphi}\right).
\]
Integrating this inequality over $\manifold$ and using 
$(1/\psi)(\varphi - \psi)^2\leq \|1/\psi\|_{\infty}(\varphi - \psi)^2$ 
yields \eqref{eq:KL_upper_bound1}.
\end{proof}

\begin{proposition}
Let $s>d/2+1$,  $\shift>0$, and assume that $\manifold$ and $\sigma$ are smooth.  
Then for every $\bmu^{\dag}\in H^s(\manifold;\mathbb{R}^d)$  there exists a 
ball $\mathfrak{B}\subset \{\bmu:\|\bmu-\bmu^{\dag}\|_{H^s}< \rho\}$ such 
that $F$ satisfies the Kullback-Leibler tangential cone condition 
\eqref{eq:KL-tcc} in $\mathfrak{B}$.  
\end{proposition}

\begin{proof}
As shown in \cite[Lemma 5.2]{HW:13}, the classical tangential cone condition 
\eqref{eq:tang_cone_Fokker} is equivalent to 
\begin{eqnarray*}
\frac{1}{C} \norm{\uF^\dag- \Opg(\altsol)}_{L^2}
- \tilde{\eta} \norm{\uF^\dag-\Opg(\sol)}_{L^2} 
&\leq \norm{\uF^\dag- \Opg(\sol) - \Opg'[\sol](\altsol-\sol)}_{L^2}\\
&\leq C \norm{\uF^\dag-\Opg(\altsol)}_{L^2} 
+ \tilde{\eta} \norm{\uF^\dag-\Opg(\sol)}_{L^2}
\end{eqnarray*}
for some constants $\tilde{\eta},C>0$ and all $\sol,\altsol\in\mathfrak{B}$.

Next we are going to show that $F$ is also continuously differentiable as a 
mapping from the H\"older space 
$C^{1,\beta}(\overline{\manifold},\mathbb{R}^d)\to L^{\infty}(\manifold)$. 
Note that the solution $u$ to \eqref{eq:stationary_weak} satisfies 
\[
\mymatrix{c c}{\tilde{L}_{\bmu} & \mathds{1}\\ \mathds{1}^{*} &0} 
\mymatrix{c}{u\\ \lambda} = \mymatrix{c}{0 \\ 1}, 
\]
where $\mathds{1}$ maps a constant $\lambda\in\mathbb{R}$ to the constant function 
with value $\lambda$ on $\manifold$, and $\mathds{1}^{*}$ is its $L^2$-adjoint. 
By Schauder estimates (see e.g.\ \cite{GT:77}) the (block-)operator 
as a mapping from $C^{2,\beta}(\overline{\manifold})\times \mathbb{R}\to 
C^{0,\beta}(\overline{\manifold})\times \mathbb{R}$ has a bounded inverse 
if $\bmu\in C^{1,\beta}(\overline{\manifold},\mathbb{R}^d)$.  
Since the block operator depends continuously and affinely linear on $\bmu$ 
in these topologies and since the operator inversion is continuously differentiable, $F$ is continuously Fr\'echet differentiable from 
$C^{1,\beta}(\overline{\manifold},\mathbb{R}^d)$ to 
$C^{2,\beta}(\overline{\manifold})$ 
and hence from $C^{1,\beta}(\overline{\manifold},\mathbb{R}^d)$ to 
$L^{\infty}(\manifold)$.  

Choose $0<\beta<s-d/2$. Then every ball $\mathfrak{B}$ 
in $H^s(\manifold)$ is compact in $C^{1,\beta}(\overline{\manifold})$, and 
the mappings $\bmu\mapsto \|F(\bmu)\|_{L^{\infty}}$ and 
$\bmu\mapsto \|F'[\bmu]\|_{C^{1,\beta}\to L^{\infty}}$ are bounded 
on $\mathfrak{B}$ as continuous functions on a compact set. 
Together with Lemma~\ref{lem:L2-KL} this implies 
\eqref{eq:KL-tcc} after possibly decreasing the radius of $\mathfrak{B}$. 
\end{proof}

\begin{proposition}
If $\mathcal{R}(\bmu)= \|\bmu\|_{H^s}^2$ with $s>d/2+1$, then 
every $\bmu^{\dag}\in H^s(\manifold;\mathbb{R}^d)$ satisfies a 
variational source condition of the form \eqref{eq:source_condition} in 
some $H^s$-ball. 
\end{proposition}

\begin{proof}
Due to the results in \cite{MH:08}, $\bmu^{\dag}$ satisfies a spectral source condition 
\[
\bmu^{\dag} = \Theta(F'[\bmu^{\dag}]^*F'[\bmu^{\dag}])w
\]
for some $w\in H^s(\manifold;\mathbb{R}^d)$ and some index function $\Theta$. 
Therefore, $\bmu^{\dag}$ also 
satisfies a variational source condition for the linear operator 
$F'[\bmu^{\dag}]$
\[
\beta D_{\calR}^{\bmu}(\bmu,\, \bmu^\dag) \le \calR(\bmu) - \calR(\bmu^\dag) + \tilde{\Lambda}\Big(\|F'[\bmu]'(\bmu-\bmu^{\dag})\|^2_{L^2}\Big) 
\]
for all $\bmu \in H^s(\manifold;\mathbb{R}^d)$ 
with another index function $\tilde{\Lambda}$ (see \cite{flemming:12b}). 
Note that the $L^2$ tangential cone condition in Theorem \ref{the:elliptic} 
implies 
\[
\|F'[\bmu^{\dag}](\bmu-\bmu^{\dag})\|_{L^2} 
\leq (1+\tilde{C}_{\bmu^{\dag}}\|\bmu^{\dag}-\bmu\|_{\infty})
\|F(\bmu)-F(\bmu^{\dag})\|_{L^2}
\]
for all $\bmu\in H^s(\manifold;\mathbb{R}^d)$. Therefore, 
$\bmu^{\dag}$ also satisfies the variational 
source condition for the nonlinear operator $F$
\[
\beta D_{\calR}^{\bmu}(\bmu,\, \bmu^\dag) \le \calR(\bmu) - \calR(\bmu^\dag) + \tilde{\Lambda}\Big(4\|\uF^\dag - F(\bmu)\|^2_{L^2}\Big) 
\]
for all $\bmu\in H^s(\manifold;\mathbb{R}^d)$ with 
$\tilde{C}_{\bmu^{\dag}}\|\bmu-\bmu^{\dag}\|_{\infty}\leq 1$. 
Together with Lemma \ref{lem:L2-KL} and the continuous embedding of 
$H^s(\manifold,\mathbb{R}^d)$ in $L^{\infty}(\manifold,\mathbb{R}^d)$ 
this entails the $\KL$ related source condition \eqref{eq:source_condition}. 
\end{proof}

To sum up, all assumptions of Theorems \ref{theo:tikh} and \ref{theo:convergencerate} are satisfied 
for our problem. It would be interesting to have explicit 
characterizations of the index function $\Lambda$ when $\bmu$ satisfies 
certain classical smoothness conditions. We intend to address this question 
in future research.

\section{Numerical simulations}\label{sec:numerics}

\subsubsection*{Implementation.}

The implementation of the iteration scheme \eqref{eq:Gauss-Newton} requires the evaluation of the forward operator $F$ and its derivative $F'$. We did this for both operators by finite elements of degree 3. The convex minimization problem which occurs in every Newton step is solved by a nested Newton iteration as described 
in \cite{HW:13}. 

In addition to the iteration \eqref{eq:Gauss-Newton} we implemented the classical Gau{\ss}-Newton method with quadratic data fidelity term. As both methods were equipped with an $H^1$-quadratic penalty term, this setup allows for a comparison of the two methods. For the latter inversion scheme the minimization problem in every Newton step becomes quadratic and can be solved by a conjugate gradient method.

\subsubsection*{Test example.}
To test the algorithm we considered a one-dimensional stochastic differential 
equation \eqref{eq:sde} with diffusion $\sigma = 0.5$ and drift
\begin{equation}\label{eq:drift}
\mu^{\dagger}(x) = -5x^3-2x-0.25 \qquad \mbox{for }x\in [-1,1], 
\end{equation}
$\mu^{\dagger}(x) = \mu^{\dagger}(1)$ for $x \ge 1$, and 
$\mu^{\dagger}(x) = \mu^{\dagger}(-1)$ for $x \le -1$. 
The drift is plotted in Figures \ref{fig:recon1} and \ref{fig:recon2}.
We simulated a path of the stochastic process with the Euler-Maruyama method on a large time interval $[0,T]$ with $T = 1000$ and with $10^5$ Euler steps. But we used only $125$ to $1000$ points in the time domain as observations of the path. This drift \eqref{eq:drift} is rather large in absolute values for $x = 1$ and $x=-1$ with a negative sign for $x=-1$. When the path jumped outside $[-1,1]$ in the simulations it jumped back into the interval in a very small number of steps. The probability to have an observation of the simulated path outside of the interval is close to 0.
To implement the forward operator we used transparent boundary conditions at $-1$ 
and $1$ as described in section \ref{sec:Fokker-Planck}. 
% \begin{center}
% \includegraphics[width=0.5\textwidth]{./exact_drift.eps}
% \captionof{figure}{Drift $\mu^{\dagger}(x)$ in \eqref{eq:drift}, which 
% is used to generate synthetic data by solving the SDE \eref{eq:sde} 
% (see Fig.~\ref{fig:path}.} 
% \end{center}

\begin{center}
\includegraphics[width=0.70\textwidth]{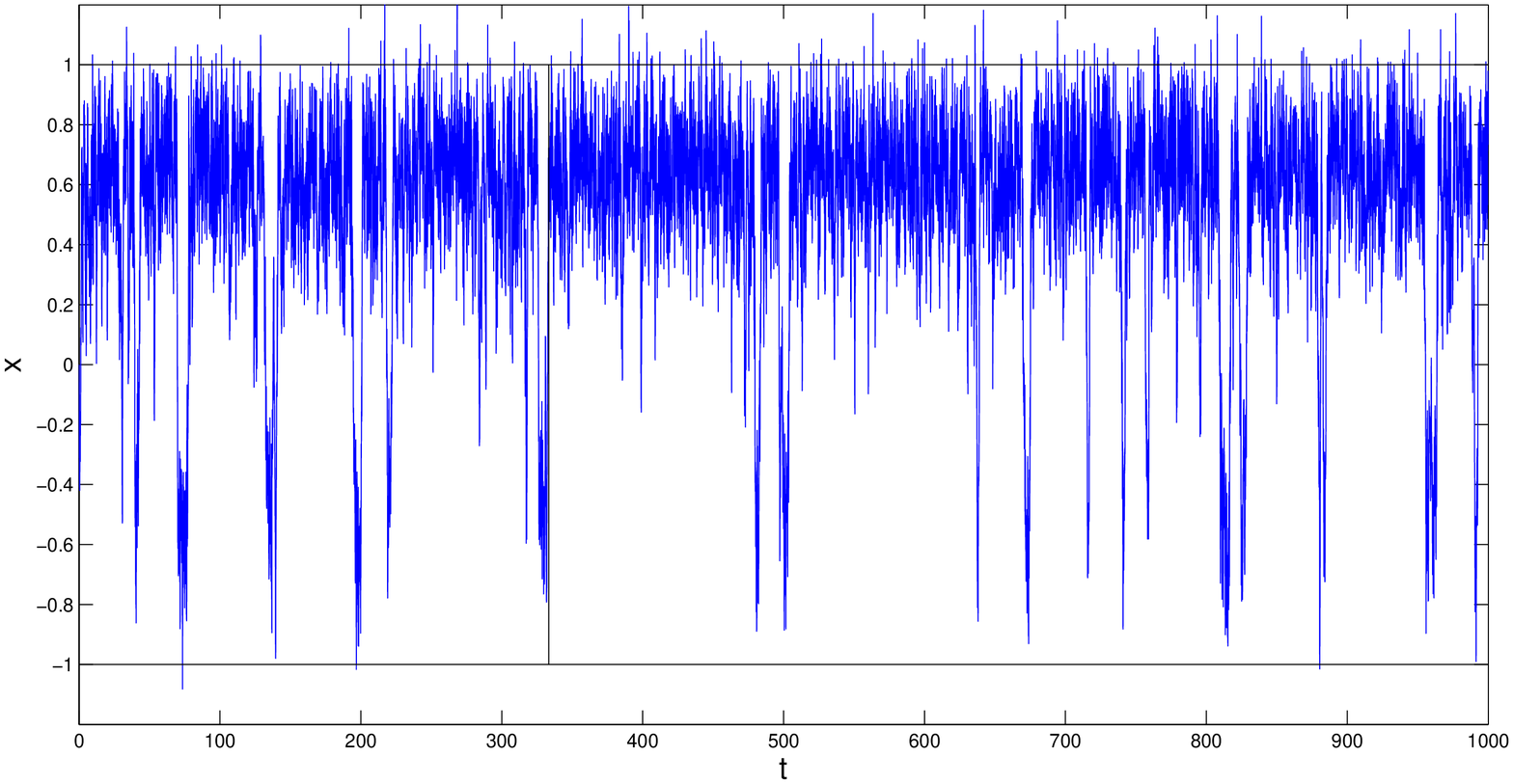} \hspace{-10mm}
\includegraphics[width=0.34\textwidth]{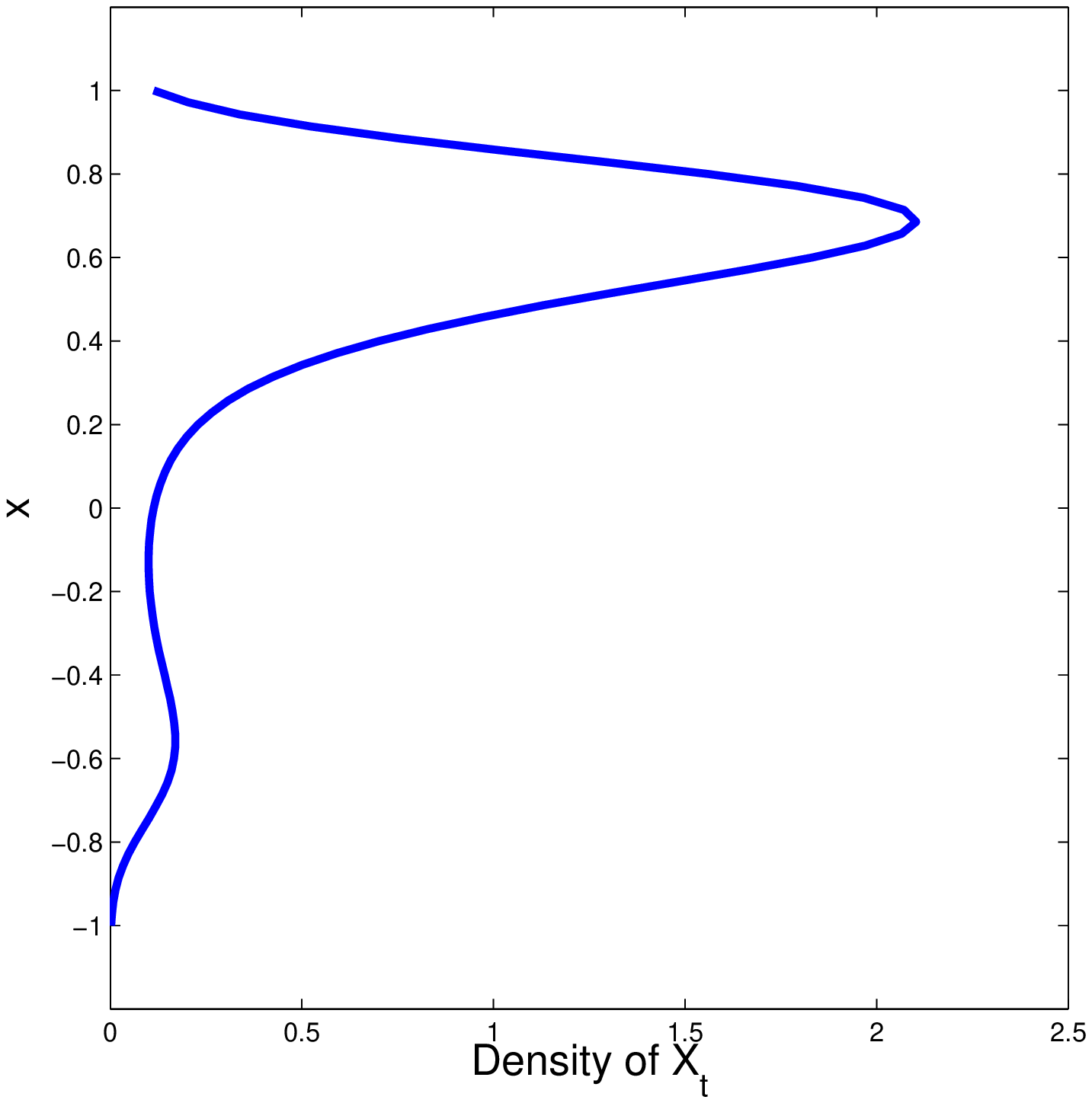}\vspace{-12pt} 
\captionof{figure}{\label{fig:path}A simulated path and the corresponding 
limit density of the process $X_t$ for $t \to \infty$.}
\end{center}

% In this test example the density of the process $X_t$ is strictly positive. Hence, the discussion in Section \ref{sec:convergence_sde} shows that the assumptions of the convergence theorem \ref{theo:convergencerate} are fulfilled. In all computations we used the zero function as initial guess. 

\subsubsection*{Results.}
We reconstructed the drift using 4 different numbers of observations of a path namely $125$, $250$, $500$, and $1000$ points. For each set of observations we reconstructed the drift using the iteratively regularized Newton method \eqref{eq:Gauss-Newton} with $\KL$ data fidelity term and additionally using the iteratively regularized Gau{\ss}-Newton method. In both reconstruction methods we assumed that the drift is known in semiinfinite intervals $(-\infty,-1]$ and $[1,\infty)$. 
Moreover, in order to compare both methods independent of a stopping rule, 
in both cases an oracle choice of the 
stopping index was used, i.e.\ the stopping index was chosen such that 
the average $L^2$-error was minimal. 

Due to the random error in the data, a statistic evaluation of the inversion methods is needed. For this purpose we repeated the procedure of simulating a path, drawing observations from it and conducting the estimations $1000$ times. The following histograms show the distribution of the $L^2$ error of both methods. The error is normalized in a way such that the error of the initial guess is $1$.
\begin{center}
% \includegraphics[width=0.51\textwidth]{./125KL.eps} \hspace{-6mm}
% \includegraphics[width=0.51\textwidth]{./125L2.eps} 
% \captionof{figure}{$125$ observations of one path: $L^2$ error of reconstruction with $\KL$ (left) and $L^2$ (right) data fidelity term.}
% 
% \includegraphics[width=0.51\textwidth]{./250KL.eps} \hspace{-6mm}
% \includegraphics[width=0.51\textwidth]{./250L2.eps} 
% \captionof{figure}{$250$ observations of one path: $L^2$ error of reconstruction with $\KL$ (left) and $L^2$ (right) data fidelity term.}
% 
% \includegraphics[width=0.51\textwidth]{./500KL.eps} \hspace{-6mm}
% \includegraphics[width=0.51\textwidth]{./500L2.eps} 
% \captionof{figure}{$500$ observations of one path: $L^2$ error of reconstruction with $\KL$ (left) and $L^2$ (right) data fidelity term.}
% 
% \includegraphics[width=0.51\textwidth]{./1000KL.eps} \hspace{-6mm}
% \includegraphics[width=0.51\textwidth]{./1000L2.eps} 
% \captionof{figure}{$1000$ observations of one path: $L^2$ error of reconstruction with $\KL$ (left) and $L^2$ (right) data fidelity term.}
\includegraphics[width=0.4\textwidth]{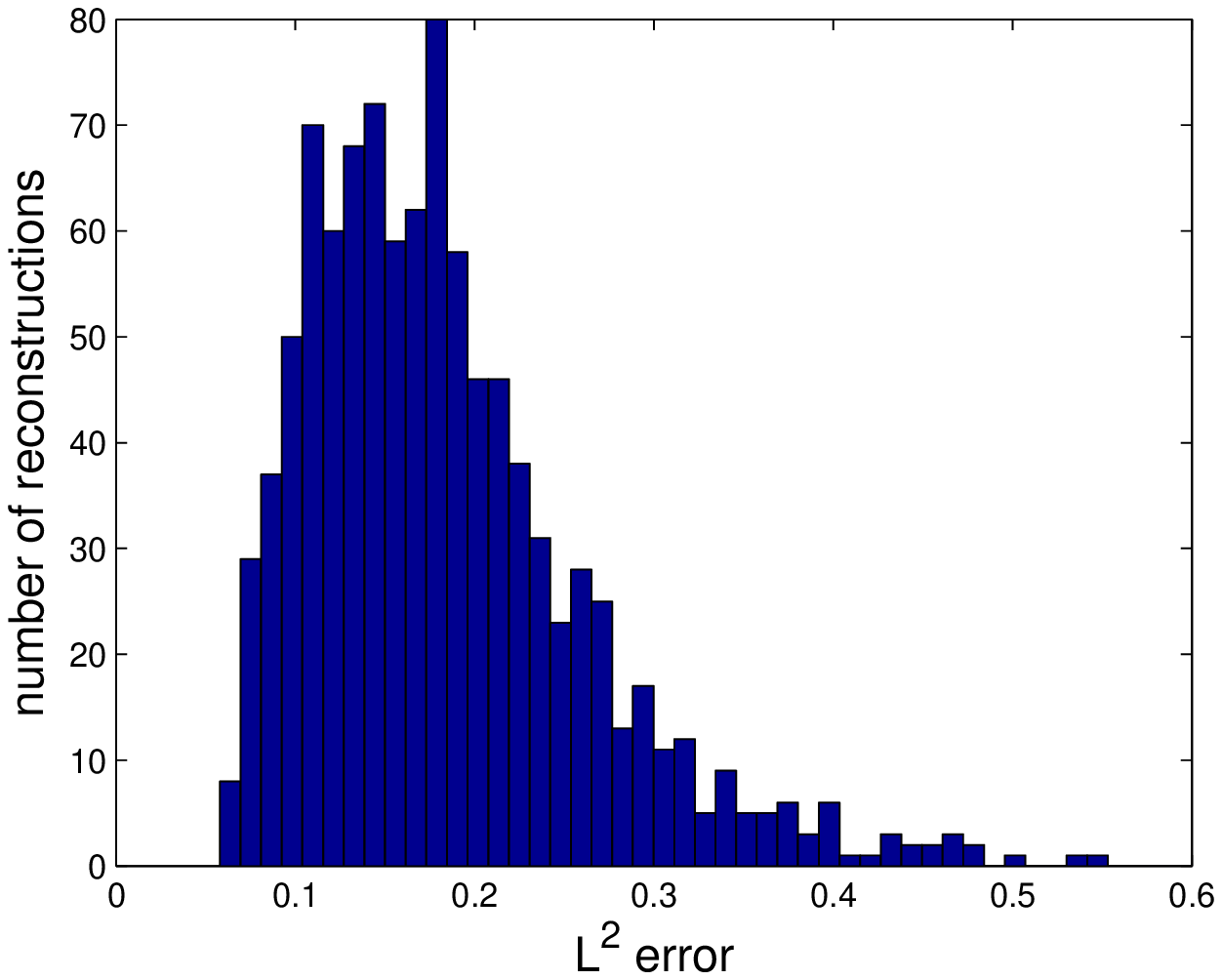}
\includegraphics[width=0.4\textwidth]{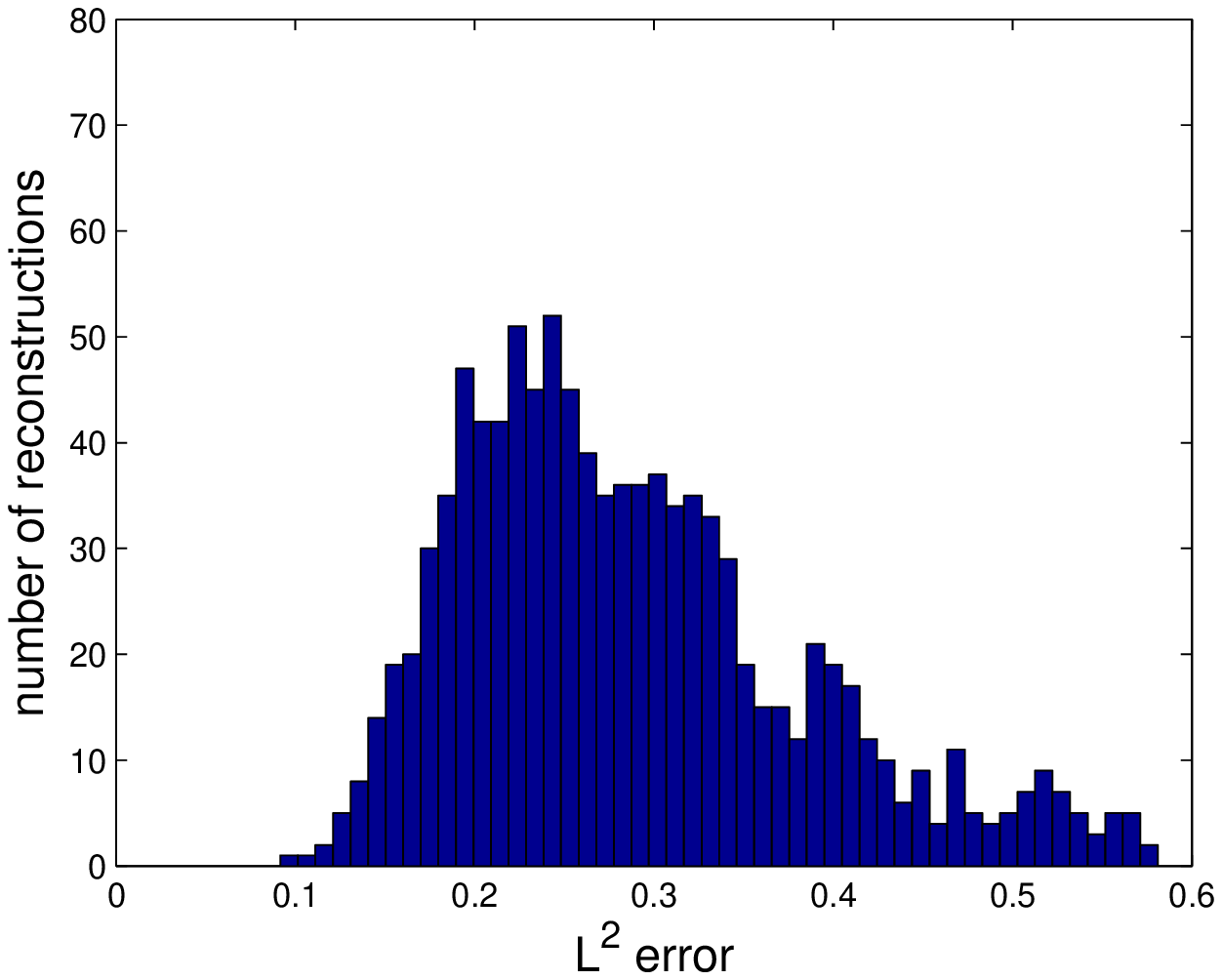}\vspace{-12pt} 
\captionof{figure}{$125$ observations of one path: $L^2$ error of reconstructions with $\KL$ (left) and $L^2$ (right) data fidelity term.}

\includegraphics[width=0.4\textwidth]{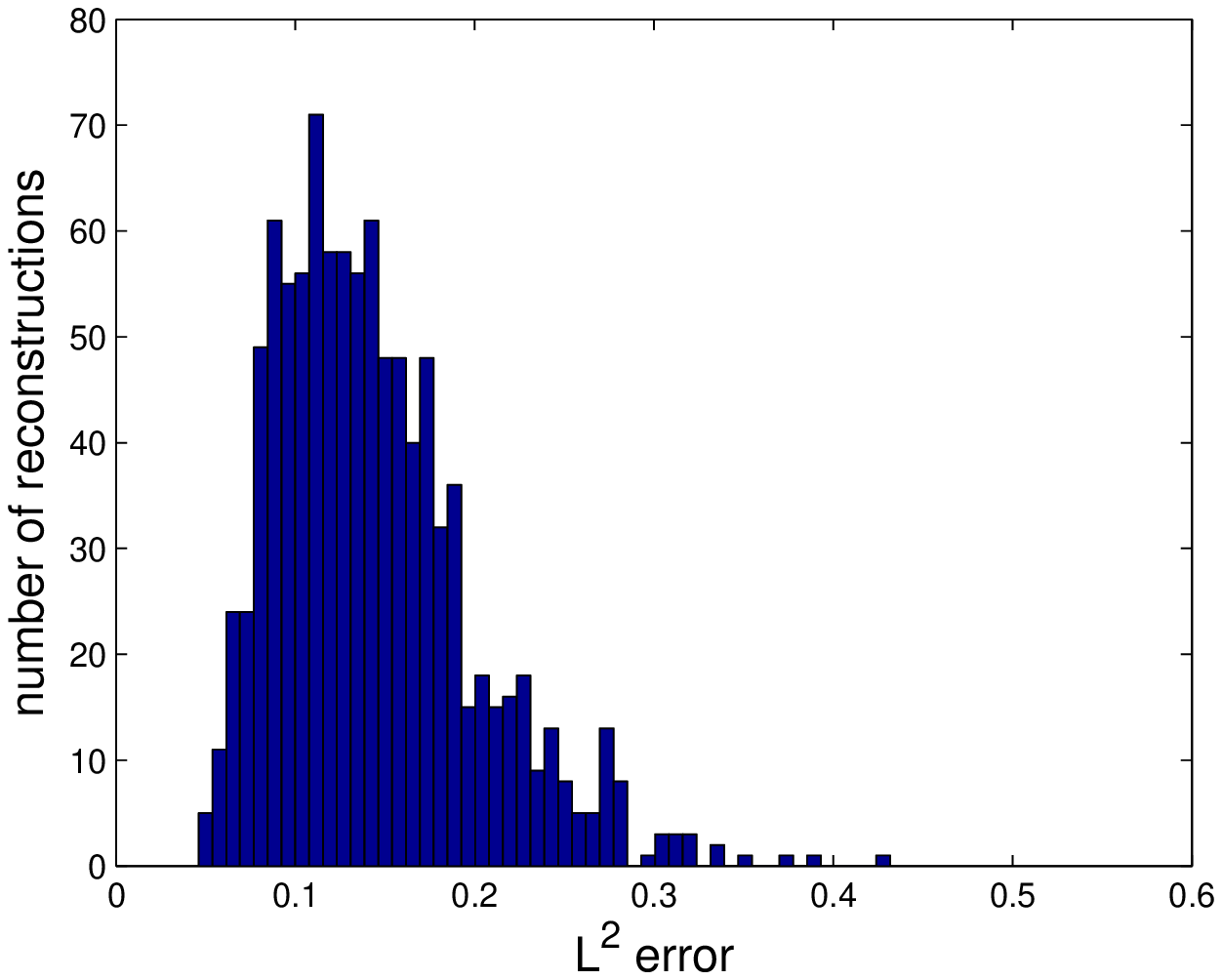}
\includegraphics[width=0.4\textwidth]{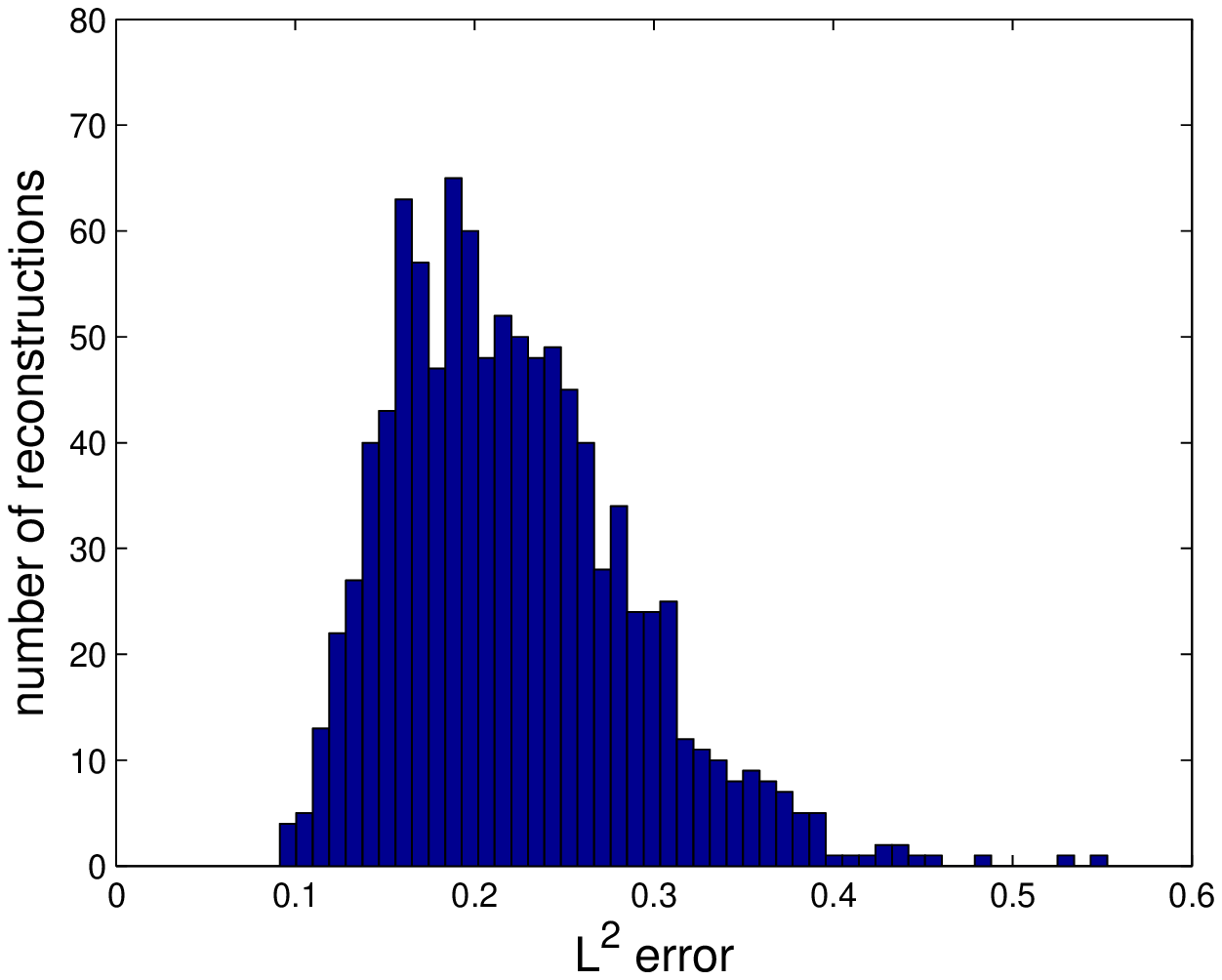} \vspace{-12pt} 
\captionof{figure}{$250$ observations of one path: $L^2$ error of reconstructions with $\KL$ (left) and $L^2$ (right) data fidelity term.}
\pagebreak
\includegraphics[width=0.4\textwidth]{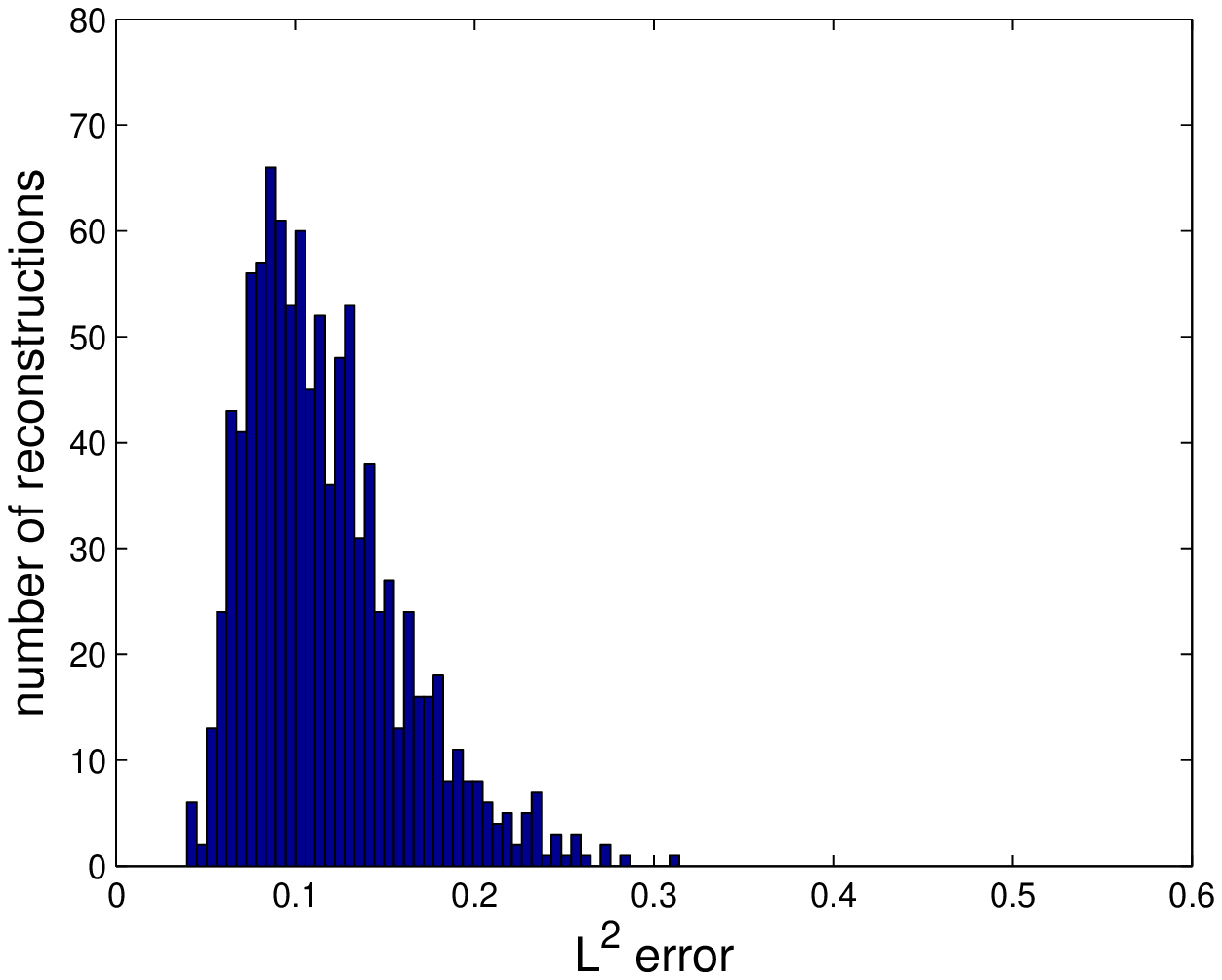}
\includegraphics[width=0.4\textwidth]{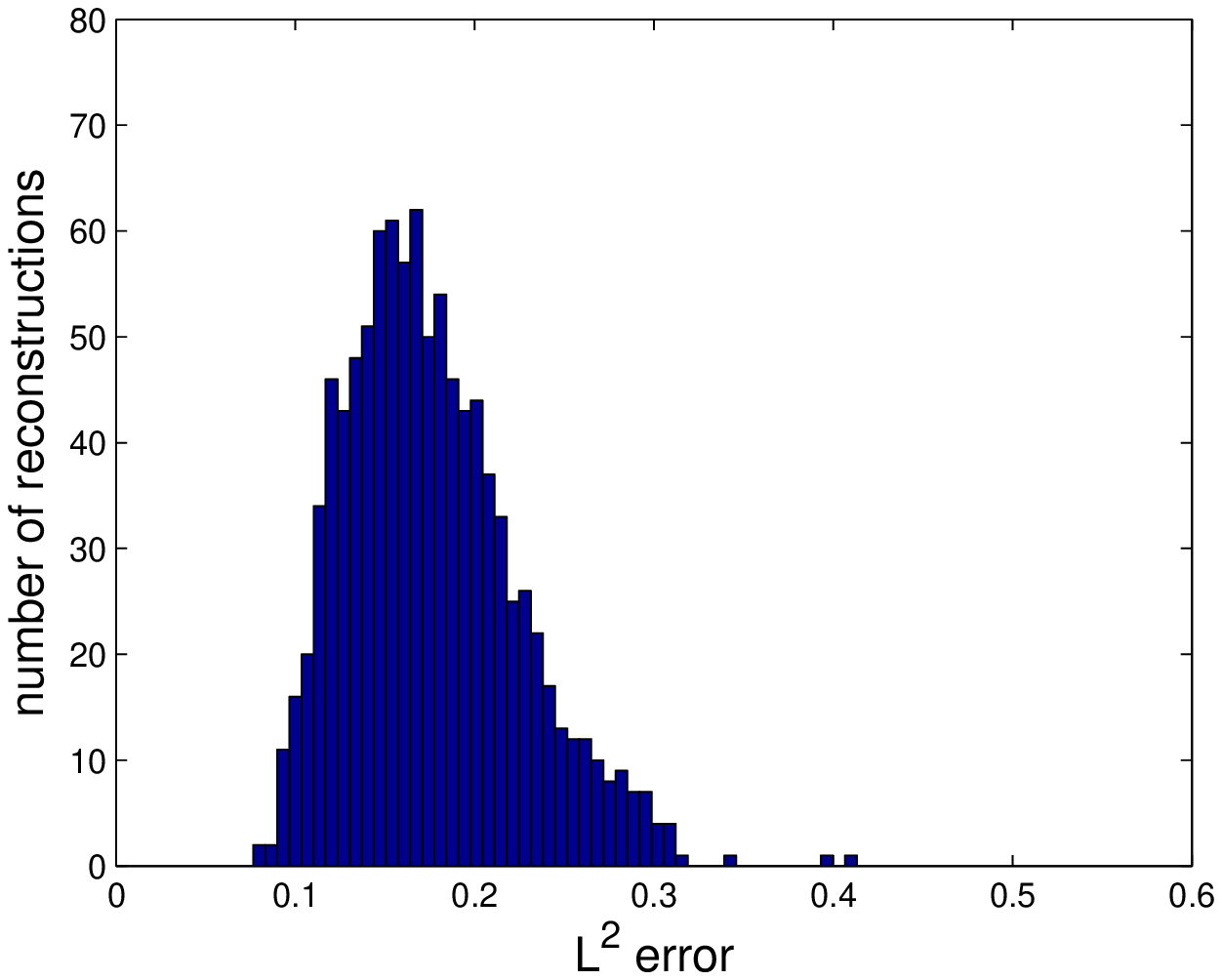} \vspace{-12pt} 
\captionof{figure}{$500$ observations of one path: $L^2$ error of reconstructions with $\KL$ (left) and $L^2$ (right) data fidelity term.}

\includegraphics[width=0.4\textwidth]{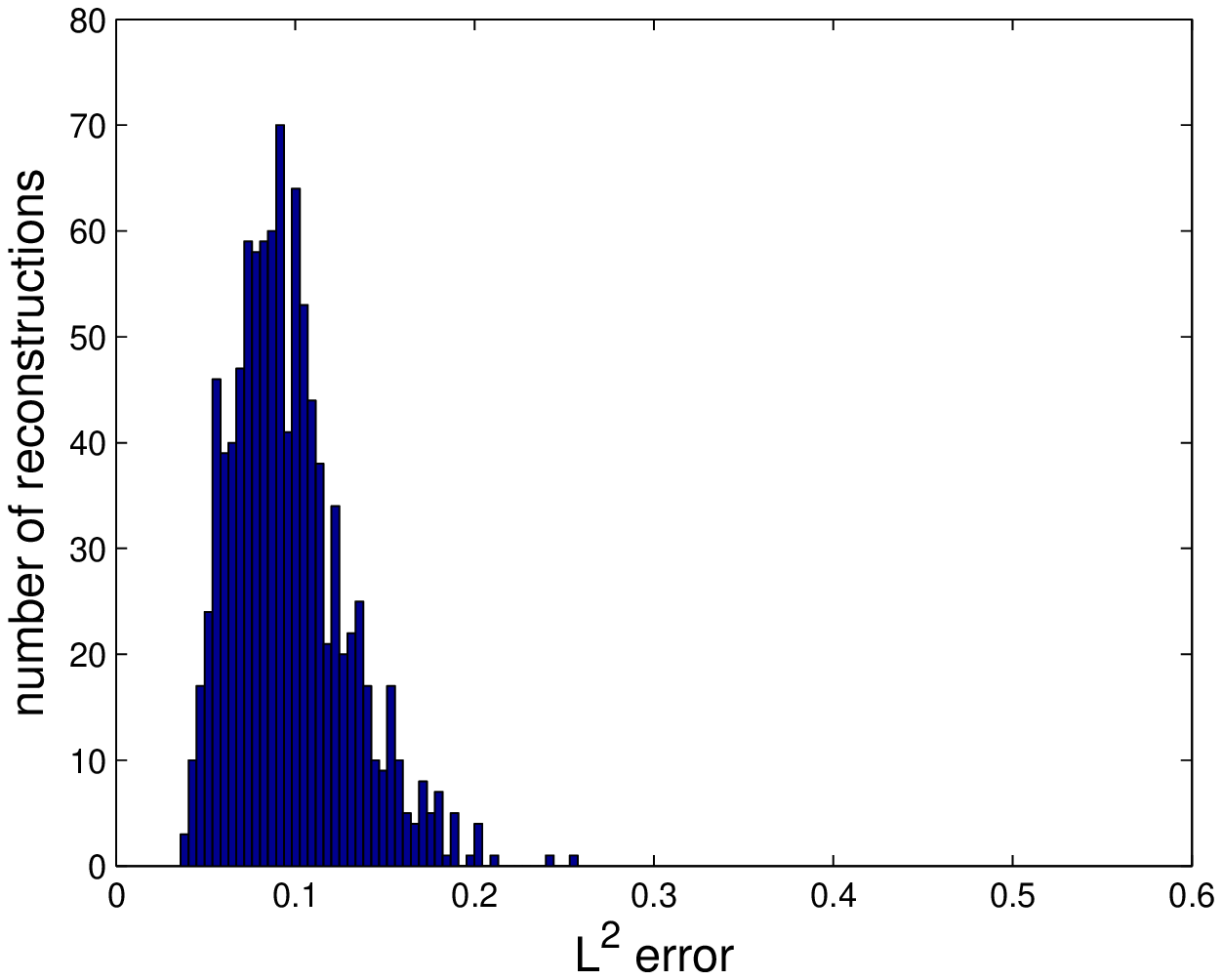}
\includegraphics[width=0.4\textwidth]{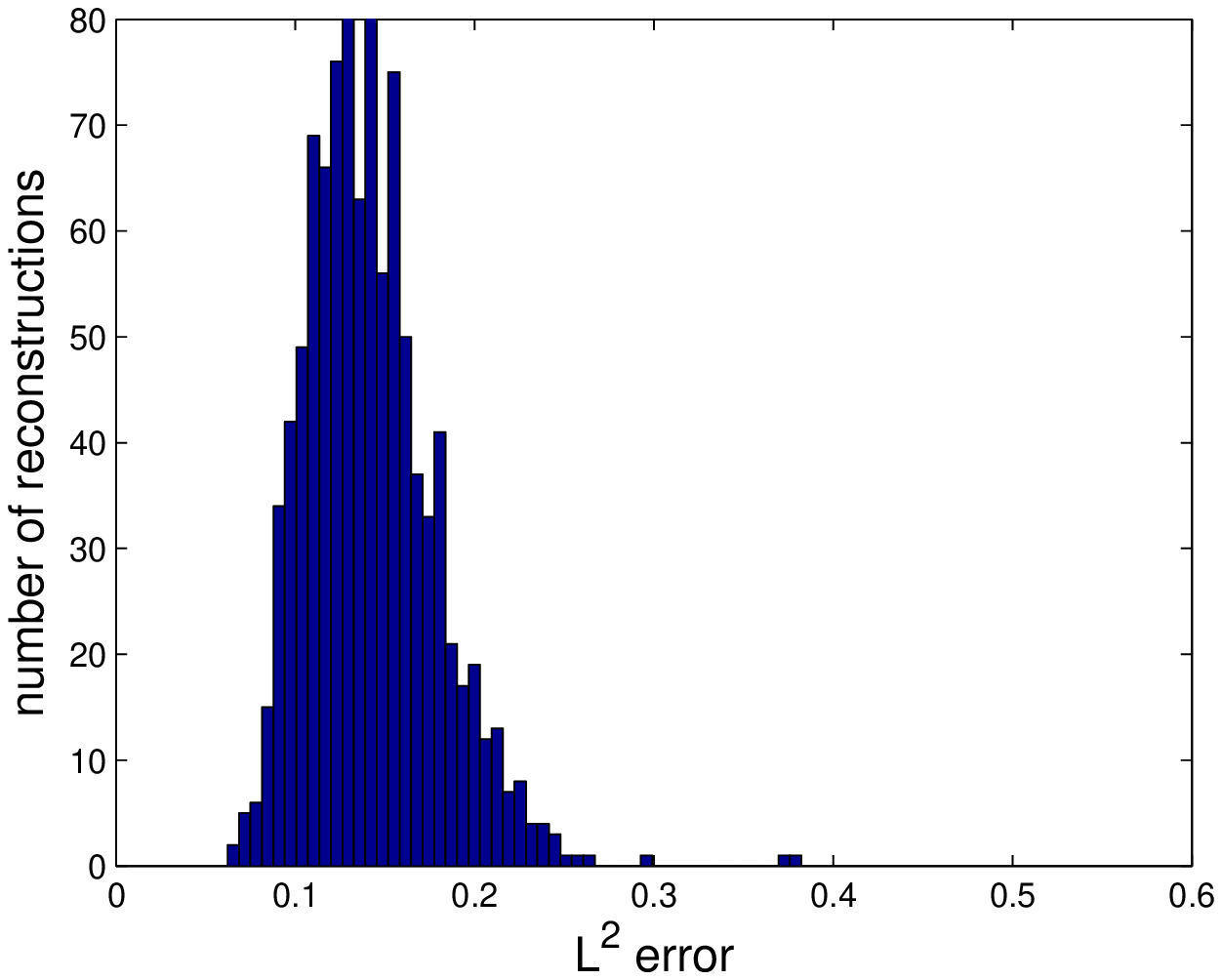} \vspace{-12pt} 
\captionof{figure}{$1000$ observations of one path: $L^2$ error of reconstructions  with $\KL$ (left) and $L^2$ (right) data fidelity term.}
\end{center}

The histograms suggest that the reconstructions with $\KL$-type data fidelity 
term have a smaller mean error and smaller variance. This is made explicit by 
the following table:
\begin{center}
\begin{tabular}{|l||l|l|l|l|}
\hline
observations & $\KL$ mean & $\KL$ variance & $L^2$ mean & $L^2$ variance\\ \hline
$125$        & 0.1832    & 0.0063        & 0.2870     & 0.0093        \\ \hline
$250$        & 0.1439    & 0.0031        & 0.2212     & 0.0044        \\ \hline
$500$        & 0.1160    & 0.0018        & 0.1759     & 0.0023        \\ \hline
$1000$       & 0.0963    & 0.0010        & 0.1417     & 0.0013        \\ \hline
\end{tabular}
\captionof{table}{Mean and variance of the error distributions when one path is observed.}
\end{center}

The following plots are typical reconstructions of the drift using a $\KL$-type data fidelity term. We chose results with a median $L^2$ error for each sample size.
\pagebreak
\begin{center}
\includegraphics[width=0.4\textwidth]{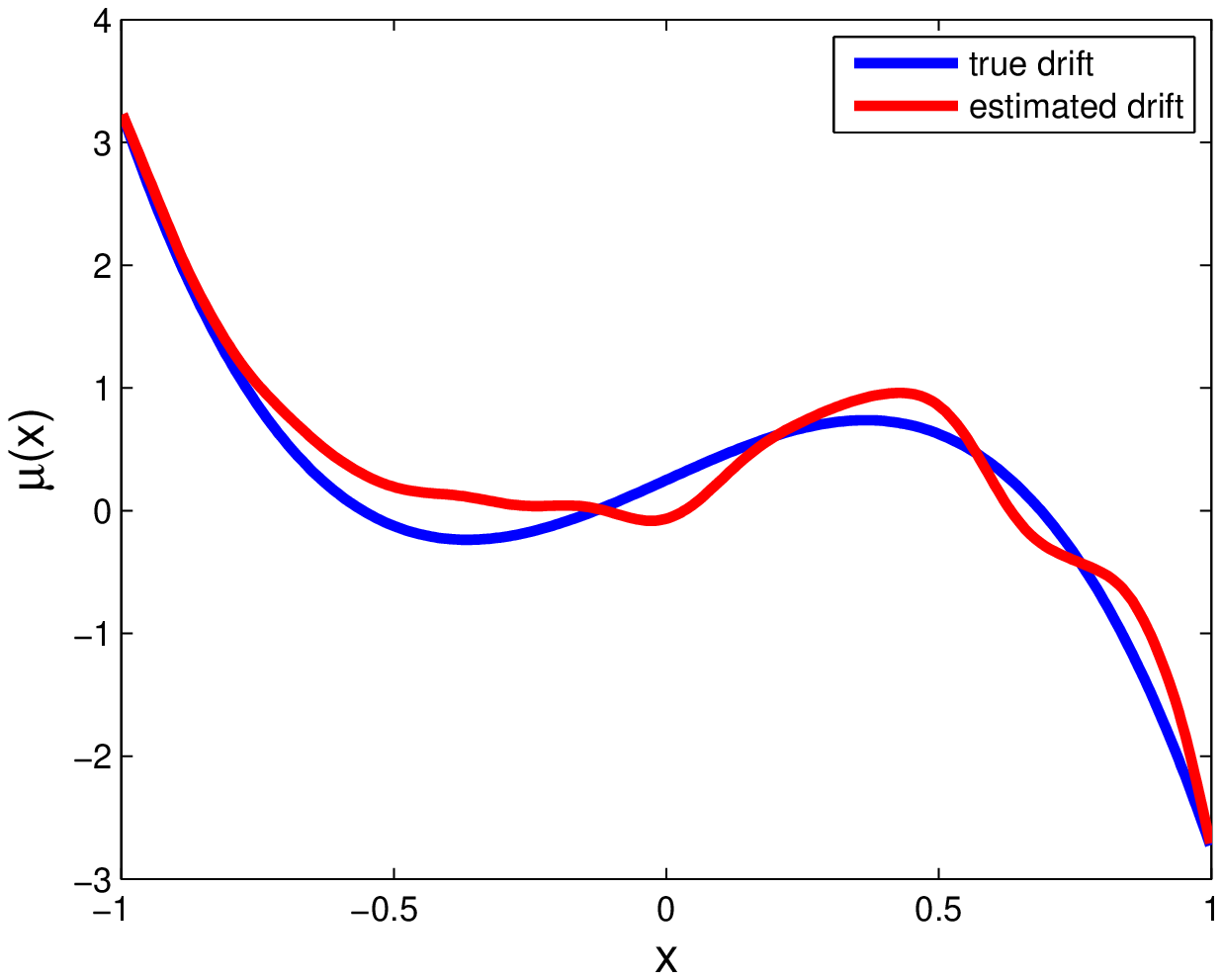} %\hspace{-6mm}
\includegraphics[width=0.4\textwidth]{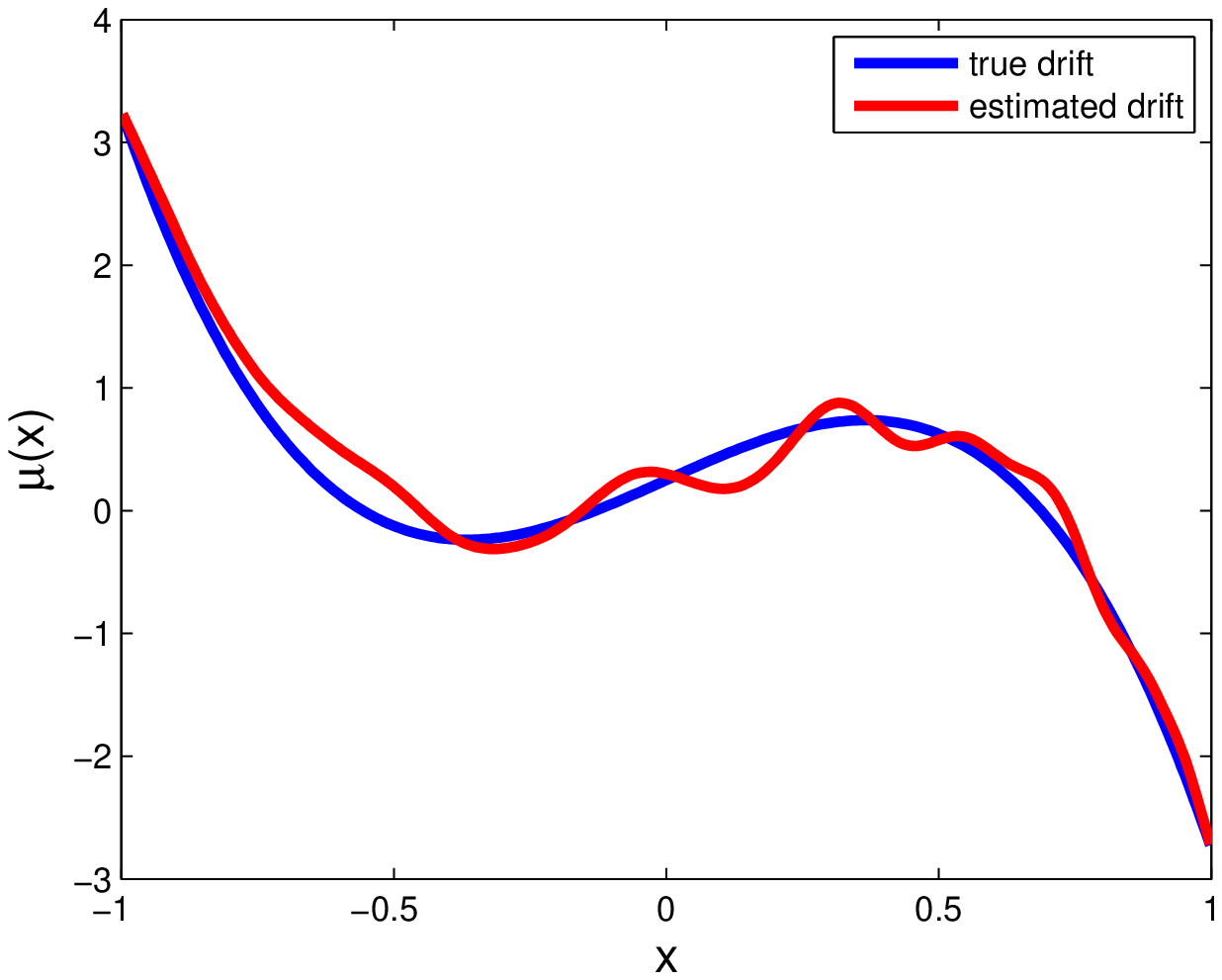} \vspace{-12pt} 
\captionof{figure}{Median reconstructions with $\KL$ data fidelity term using $125$ (left) and $250$ (right) observations of one path.}\label{fig:recon1} 

\includegraphics[width=0.4\textwidth]{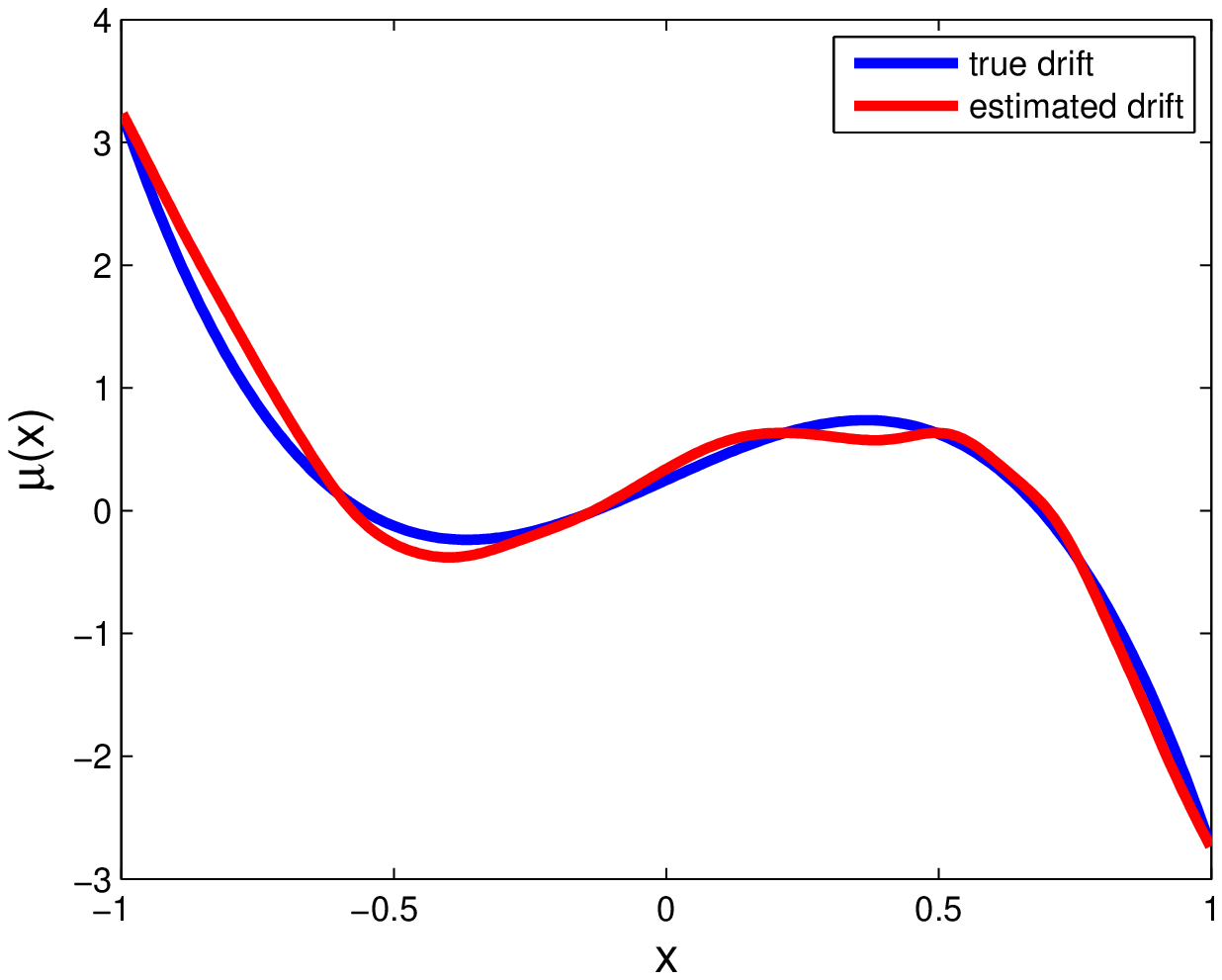} %\hspace{-6mm}
\includegraphics[width=0.4\textwidth]{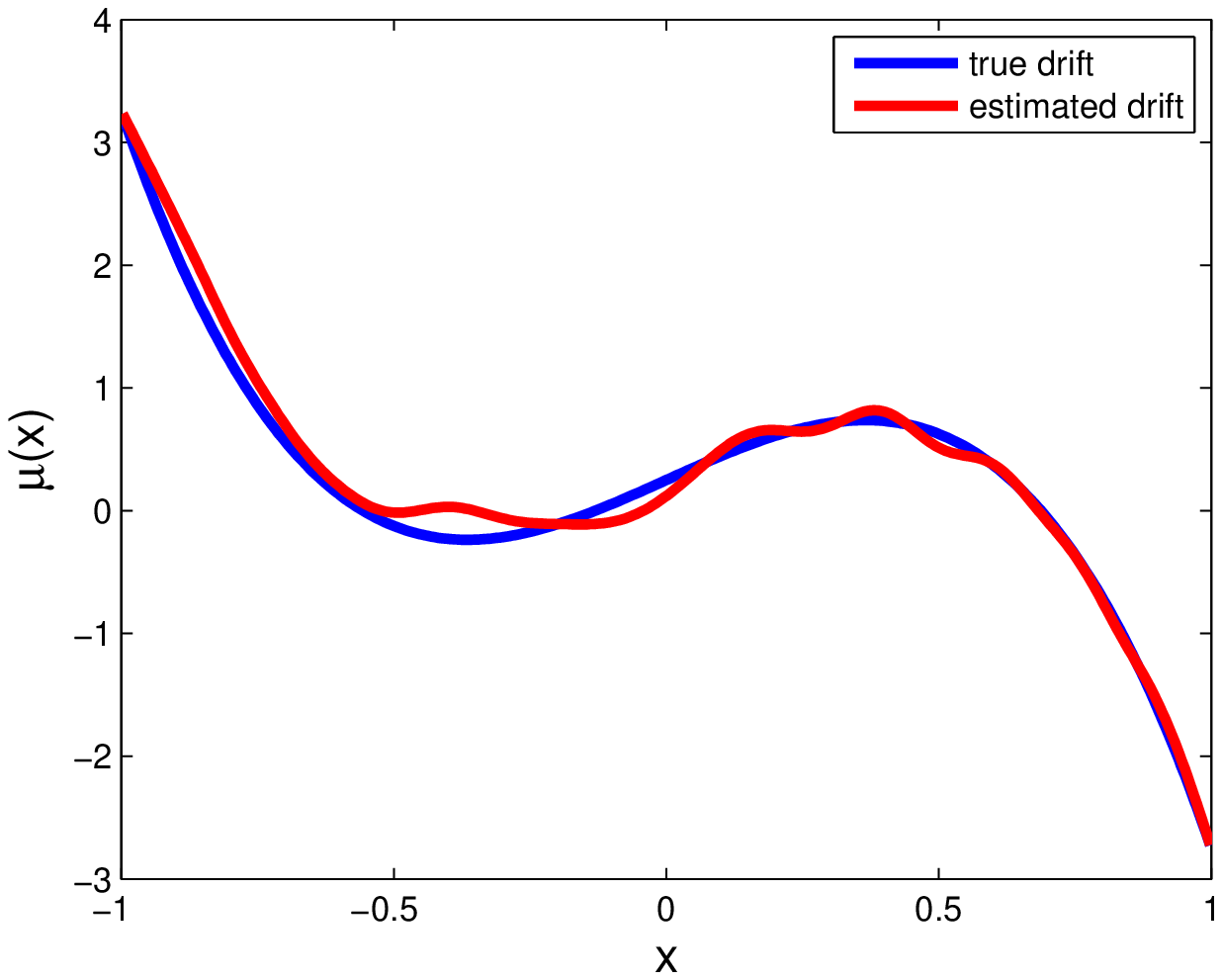} \vspace{-12pt} 
\captionof{figure}{Median reconstructions with $\KL$ data fidelity term using $500$ (left) and $1000$ (right) observations of one path.}\label{fig:recon2}
\end{center}

We summarize that in our numerical simulations the iteratively regularized Newton method with $\KL$-type or with $L^2$ data fidelity term works well as nonparametric estimator of the drift coefficient. Reduction of mean and variance of the $L^2$ error with increasing number of data is observable. The advantage of a $\KL$-type 
data fidelity term is a significantly smaller mean and variance of the $L^2$ error compared to the inversion with $L^2$ data fidelity term.

\subsubsection*{Modifications of the setup.}
In addition to the systematic numerical study above we tested the inversion scheme in two modified setups. The first variation of the setting above is to assume that the 
true values of the drift for $x\geq 1$ and $x\leq-1$ are unknown. 
Naturally, this makes the estimation of the drift close to the boundary more difficult. In addition, observations in this regions are rare in our examples as can be seen in the limit density of the process. Furthermore, the values of the drift at the boundaries are rather large in absolute values which amplifies the problem. The following plots show typical reconstructions in this case.
\begin{center}
\includegraphics[width=0.35\textwidth]{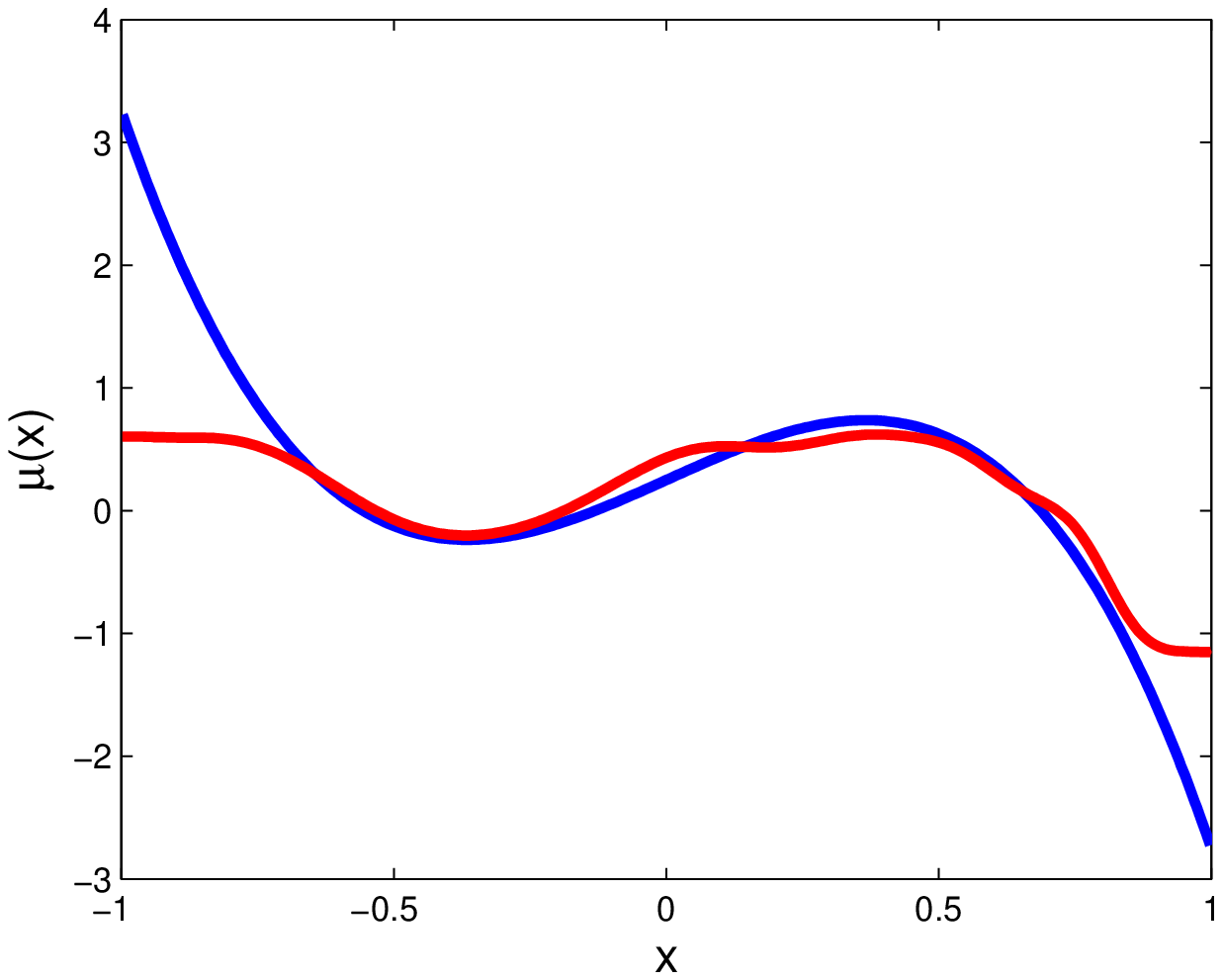}\hspace{-6mm}
\includegraphics[width=0.35\textwidth]{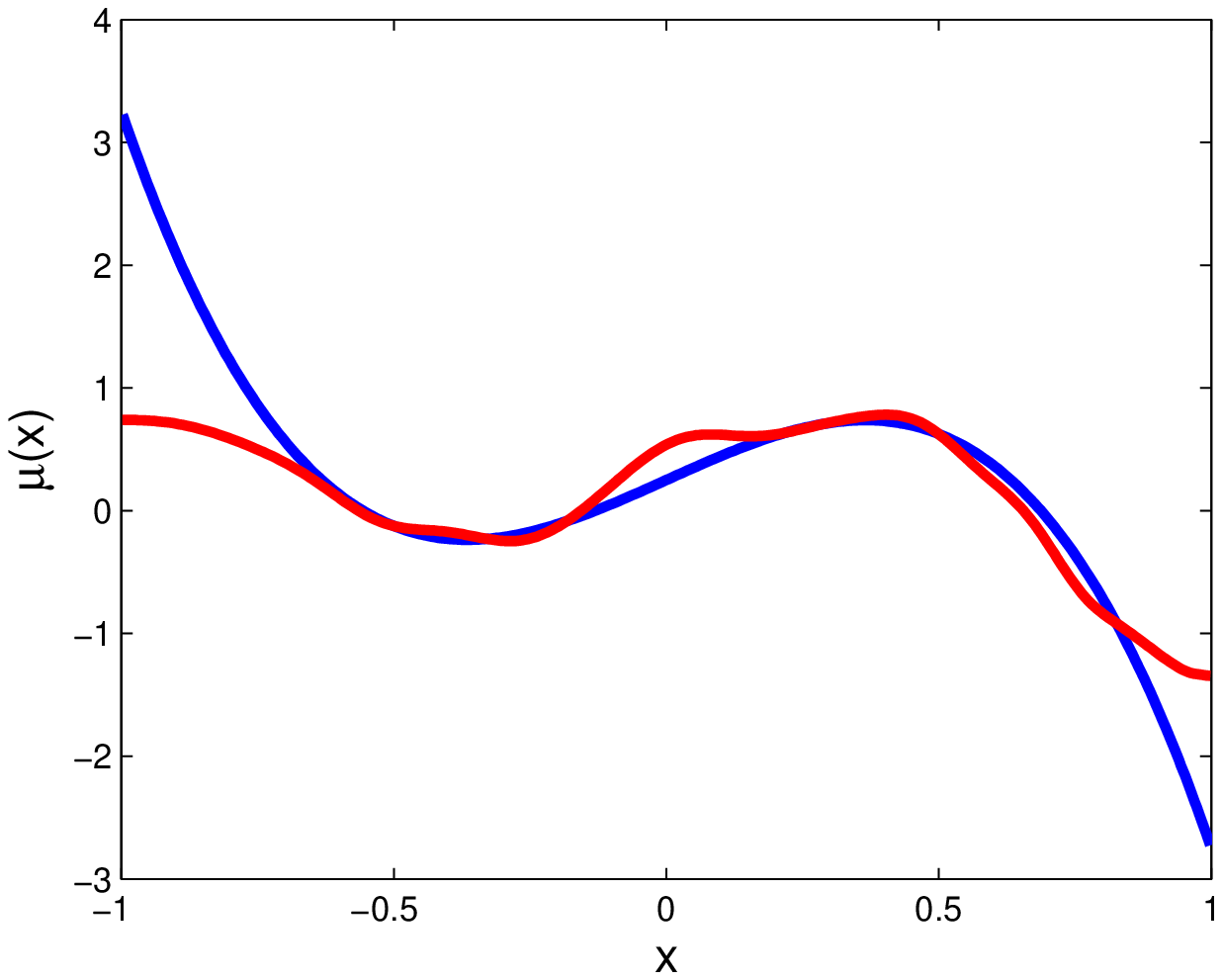} \hspace{-6mm}
\includegraphics[width=0.35\textwidth]{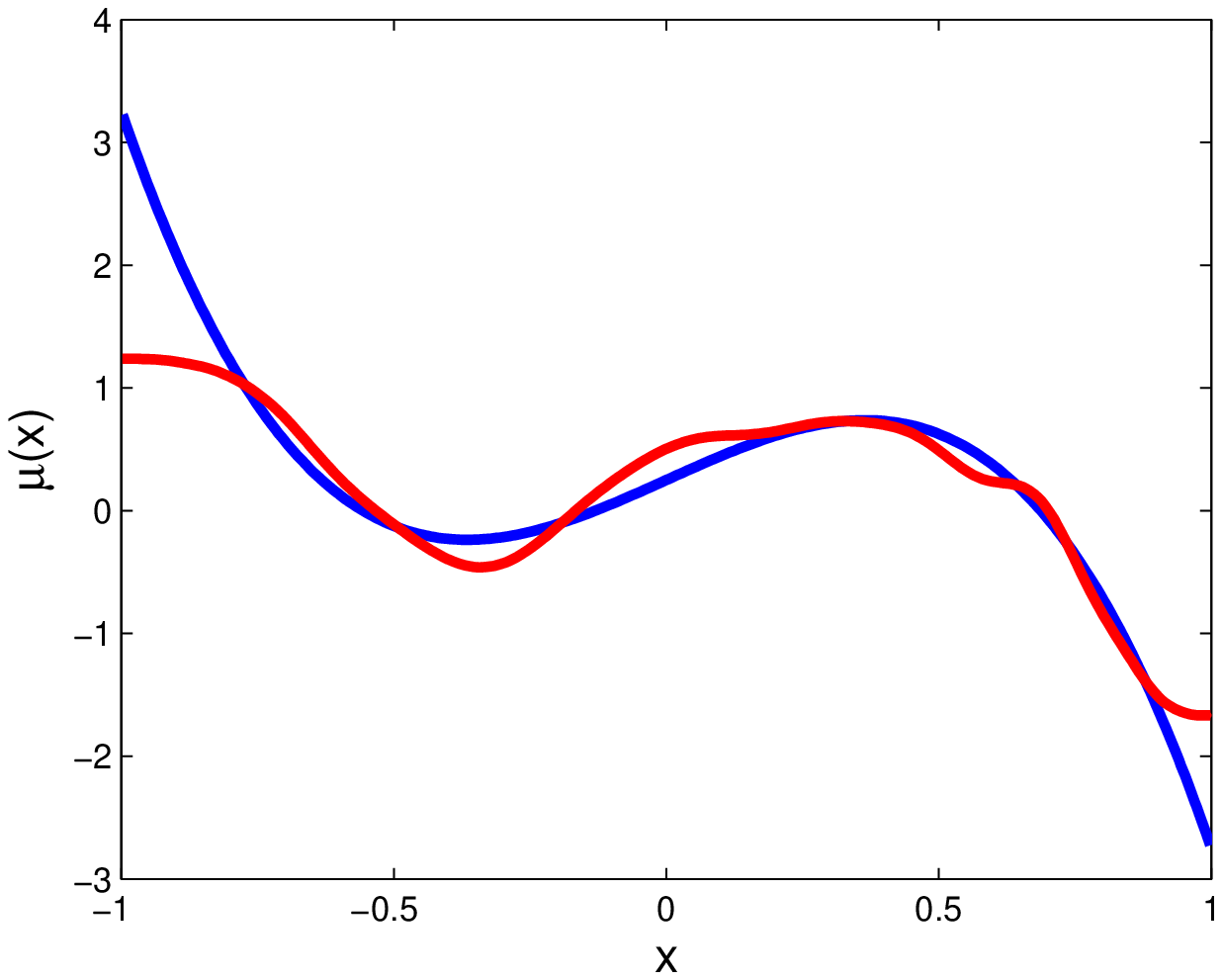}\vspace{-12pt} 
\captionof{figure}{Reconstructions with $\KL$-type data fidelity term using $250$ (left), $500$ (middle), and 1000 (right) observations of one path. red -- reconstruction, blue -- true drift}
\end{center}

As a second modification of the setup we implemented the first scenario discussed in 
the introduction. I.e.\ we simulated many paths with common starting point over 
a smaller period of time instead of simulating one path over a long period of time. 
Each path is observed at one single time point $T$. 
The operator $F$ must be modified for this setting. Instead of solving the 
elliptic problem \eqref{eq:fokker-planck_stationary_unbounded} we have to solve the parabolic problem \eqref{eq:fokker-planck} in each Newton step. We implemented this by finite elements of order three together with an implicit Euler scheme. The following plots show examples for simulated paths on the time interval $[0,1]$, the density of the process $X_t$, and reconstructions of the drift. All paths start at $0$ and observations where made at $T=1$. As above we assumed that the boundary values of the drift are unknown. 
\begin{center}
\includegraphics[width=0.45\textwidth]{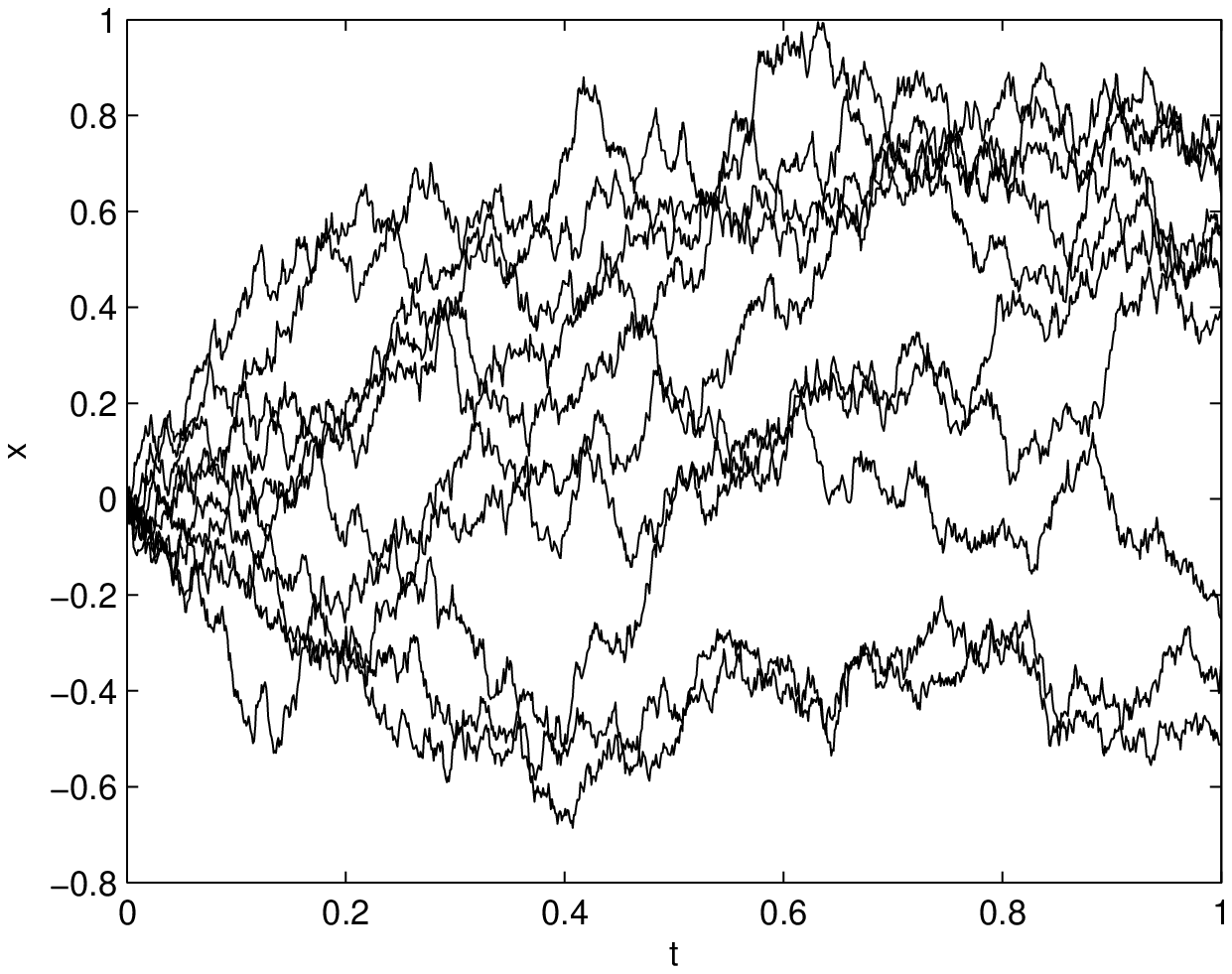}
\includegraphics[width=0.47\textwidth]{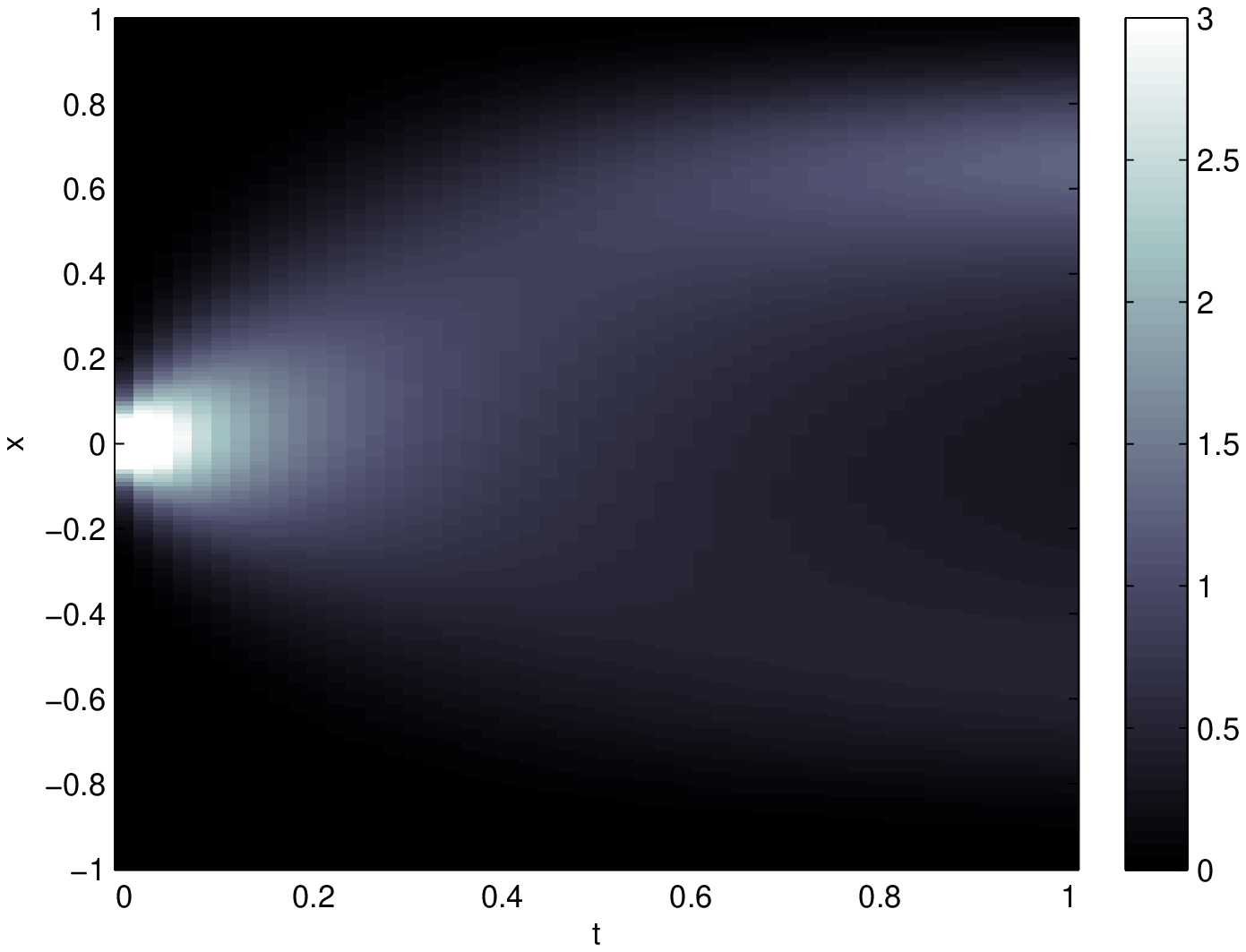}\vspace{-12pt} 
\captionof{figure}{\label{fig:parabolic_illu}10 simulated paths (left), density of $X_t$ (right)} 
\includegraphics[width=0.35\textwidth]{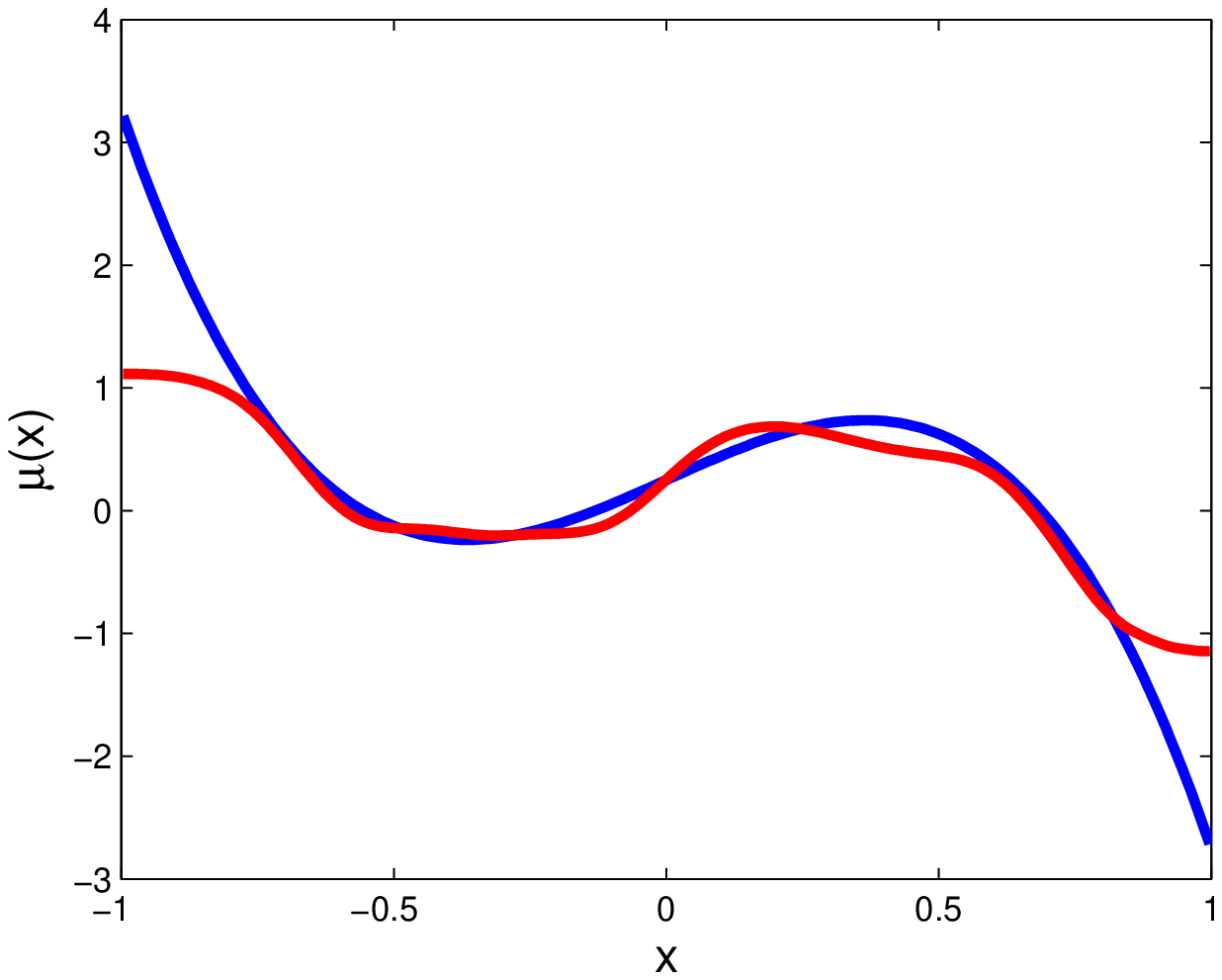}\hspace{-6mm}
\includegraphics[width=0.35\textwidth]{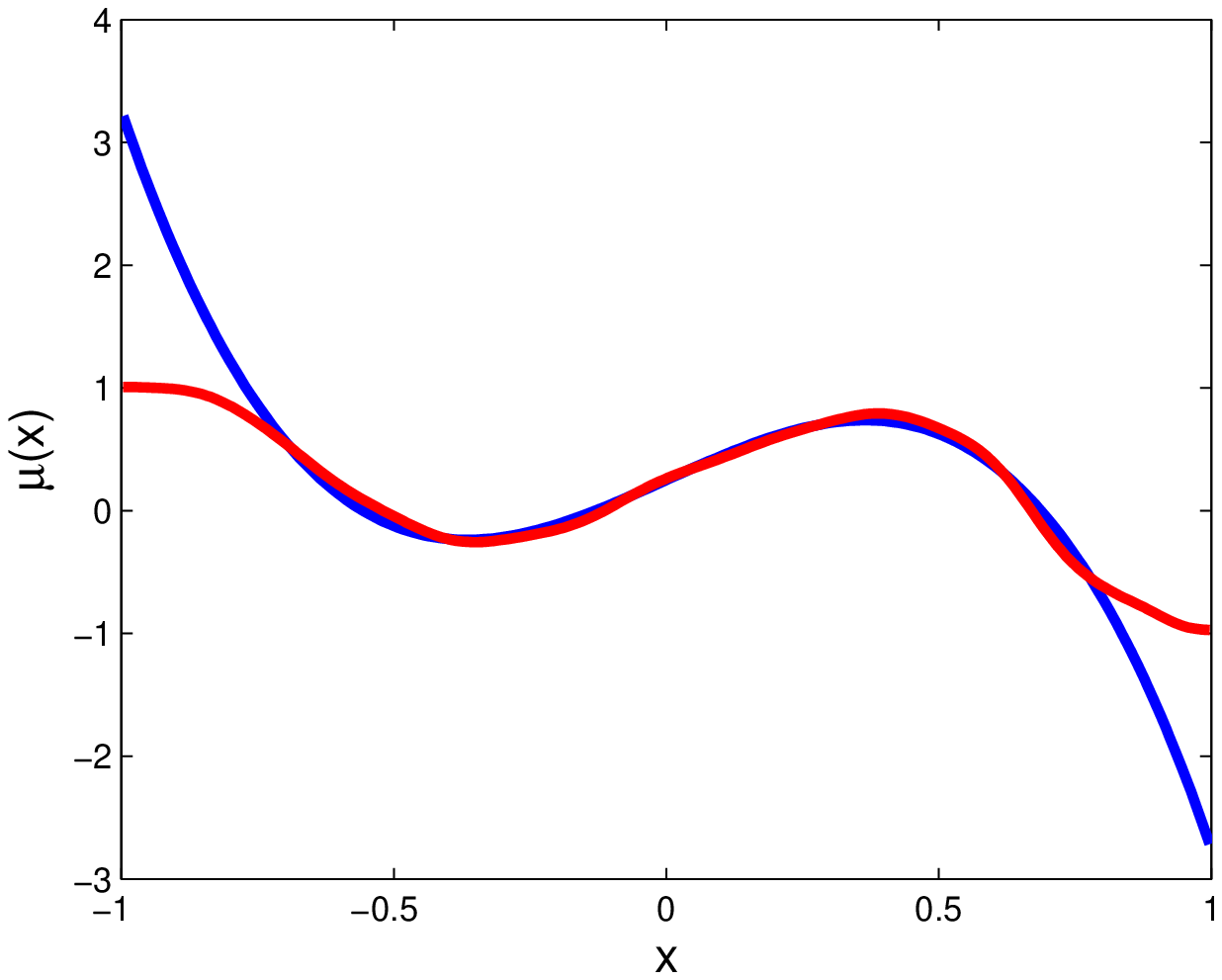} \hspace{-6mm}
\includegraphics[width=0.35\textwidth]{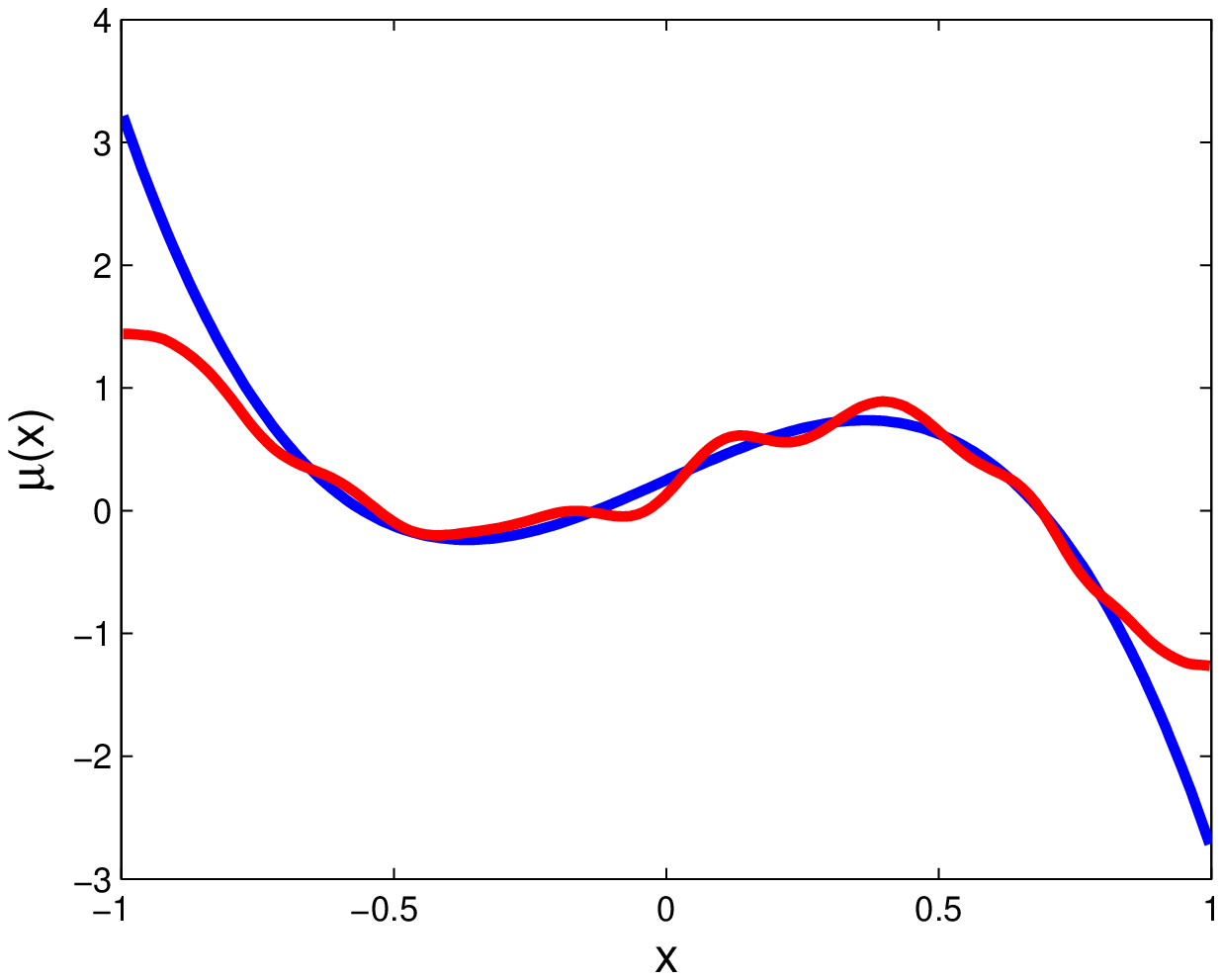}\vspace{-12pt} 
\captionof{figure}{Reconstructions with $\KL$-type data fidelity term using $250$ (left), $500$ (middle), and 1000 (right) simulated paths. red -- reconstruction, blue -- true drift. The setting is illustrated in Fig.~\ref{fig:parabolic_illu}.}
\end{center}
We can conclude that the algorithm works well in the modified setups. The problems with estimation close to the boundary are typical in nonparametric methods. Furthermore, our test examples are particularly prone to these problems. Nevertheless, our algorithm produces good results in the interior of the interval 
in these cases. 

\section{Conclusions}
We presented general convergence rate results for estimating parameters in stochastic 
differential equations by variational regularization methods 
using Kullback-Leibler-type data fidelity terms. Such terms naturally 
appear as negative log-likelihood functionals if the observations 
of paths are described by independent identically distributed random variables. 
An advantage of this approach is its flexibility. For example, it can also be used 
to estimate the volatility, initial conditions or coefficients in boundary conditions, 
and it can handle observations only in part of the domain, observations of many 
paths at many times, and observations of a whole Markov operator. However, in each 
situation the conditions of our convergence theorems have to be checked, which may 
not always be an easy task. 

Here we showed that the assumptions of our general convergence theorems are fulfilled  
for the estimation of the drift in arbitrary space dimensions. A more explicit 
characterization of the conditions for rates of convergence would be desirable, but 
has to be left for future research. 

We demonstrated by Monte-Carlo experiments that Kullback-Leibler-type data fidelity 
terms yield significantly better results than quadratic data fidelity terms. 
%which are easier to implement and to analyze. 

\ack
The authors would like to thank Christian Bender and Thomas Schuster for helpful discussions. 
Financial support by German Research Foundation DFG through the German-Swiss Research Group 
FOR 916 is gratefully acknowledged. 

\section*{References}
\bibliography{ref}

\begin{thebibliography}{10}

\bibitem{bakush:92}
A.~B. Bakushinski{\u\i}.
\newblock On a convergence problem of the iterative-regularized
  {G}auss-{N}ewton method.
\newblock {\em Zhurnal Vychislitel’noi Matematiki i Matematicheskoi Fiziki},
  32(9):1503--1509, 1992.

\bibitem{BK:04b}
A.~B. Bakushinski{\u\i} and M.~Y. Kokurin.
\newblock {\em Iterative Methods for Approximate Solution of Inverse Problems}.
\newblock Springer, Dordrecht, 2004.

\bibitem{BB:11}
M.~Benning and M.~Burger.
\newblock Error estimates for general fidelities.
\newblock {\em Electron. Trans. Numer. Anal.}, 38:44--68, 2011.

\bibitem{BS:73}
F.~Black and M.~S. Scholes.
\newblock The pricing of options and corporate liabilities.
\newblock {\em Journal of Political Economy}, 81(3):637--654, 1973.

\bibitem{BNS:97}
B.~Blaschke, A.~Neubauer, and O.~Scherzer.
\newblock On convergence rates for the iteratively regularized {G}auss-{N}ewton
  method.
\newblock {\em IMA Journal of Numerical Analysis}, 17:421--436, 1997.

\bibitem{BL:91}
J.~M. Borwein and A.~S. Lewis.
\newblock Convergence of best entropy estimates.
\newblock {\em SIAM J. Optim.}, 1(2):191--205, 1991.

\bibitem{BO:04}
M.~Burger and S.~Osher.
\newblock Convergence rates of convex variational regularization.
\newblock {\em Inverse Problems}, 20(5):1411--1421, 2004.

\bibitem{DZ:13}
A.~D. Cezaro and J.~P. Zubelli.
\newblock The tangential cone condition for the iterative calibration of local
  volatility surfaces.
\newblock {\em IMA J. Applied Math.}, 2013.

\bibitem{CGCR:07}
F.~Comte, V.~Genon-Catalot, and Y.~Rozenholc.
\newblock Penalized nonparametric mean square estimation of the coefficients of
  diffusion processes.
\newblock {\em Bernoulli}, 13(2):514--543, 2007.

\bibitem{Crepey:031}
S.~Cr{\'e}pey.
\newblock Calibration of the local volatility in a generalized
  {B}lack-{S}choles model using {T}ikhonov regularization.
\newblock {\em SIAM J. Math. Anal.}, 34(5):1183--1206 (electronic), 2003.

\bibitem{Crepey:032}
S.~Cr{\'e}pey.
\newblock Calibration of the local volatility in a trinomial tree using
  {T}ikhonov regularization.
\newblock {\em Inverse Problems}, 19(1):91--127, 2003.

\bibitem{CSZ:12}
A.~De~Cezaro, O.~Scherzer, and J.~P. Zubelli.
\newblock Convex regularization of local volatility models from option prices:
  convergence analysis and rates.
\newblock {\em Nonlinear Anal.}, 75(4):2398--2415, 2012.

\bibitem{DV:10}
J.~Droniou and J.-L. V{\'a}zquez.
\newblock Noncoercive convection-diffusion elliptic problems with {N}eumann
  boundary conditions.
\newblock {\em Calc. Var. Partial Differential Equations}, 34(4):413--434,
  2009.

\bibitem{DFHJM:13}
F.~Dunker, J.-P. Florens, T.~Hohage, J.~Johannes, and E.~Mammen.
\newblock Iterative estimation of solutions to noisy nonlinear operator
  equations in nonparametric instrumental regression.
\newblock {\em Journal of Econometrics}, 178:444--455, 2014.

\bibitem{EggerEngl:05}
H.~Egger and H.~W. Engl.
\newblock Tikhonov regularization applied to the inverse problem of option
  pricing: convergence analysis and rates.
\newblock {\em Inverse Problems}, 21(3):1027--1045, 2005.

\bibitem{eggermont:93}
P.~P.~B. Eggermont.
\newblock Maximum entropy regularization for {F}redholm integral equations of
  the first kind.
\newblock {\em SIAM J. Math. Anal.}, 24:1557--1576, 1993.

\bibitem{Flemming:10}
J.~Flemming.
\newblock Theory and examples of variational regularization with non-metric
  fitting functionals.
\newblock {\em J. Inverse Ill-Posed Probl.}, 18(6):677--699, 2010.

\bibitem{flemming:12b}
J.~Flemming.
\newblock {\em Generalized {T}ikhonov regularization and modern convergence
  rate theory in {B}anach spaces}.
\newblock Shaker Verlag, Aachen, 2012.

\bibitem{GT:77}
D.~Gilbarg and N.~S. Trudinger.
\newblock {\em Elliptic Partial Differential Equations of Second Order}.
\newblock Springer, 1977.

\bibitem{HNS:95}
M.~Hanke, A.~Neubauer, and O.~Scherzer.
\newblock A convergence analysis of the {L}andweber iteration for nonlinear
  ill-posed problems.
\newblock {\em Numer. Math.}, 72:21--37, 1995.

\bibitem{Hoffmann:99}
M.~Hoffmann.
\newblock Adaptive estimation in diffusion processes.
\newblock {\em Stochastic Process. Appl.}, 79(1):135--163, 1999.

\bibitem{HKPS:08}
B.~Hofmann, B.~Kaltenbacher, C.~P{\"o}schl, and O.~Scherzer.
\newblock A convergence rates result for {T}ikhonov regularization in {B}anach
  spaces with non-smooth operators.
\newblock {\em Inverse Problems}, 23(3):987--1010, 2007.

\bibitem{hohage:97}
T.~Hohage.
\newblock Logarithmic convergence rates of the iteratively regularized
  {G}auss-{N}ewton method for an inverse potential and an inverse scattering
  problem.
\newblock {\em Inverse Problems}, 13:1279--1299, 1997.

\bibitem{HW:13}
T.~Hohage and F.~Werner.
\newblock Iteratively regularized {N}ewton-type methods for general data misfit
  functionals and applications to {P}oisson data.
\newblock {\em Numer. Math.}, 123(4):745--779, 2013.

\bibitem{HJL:07}
A.~Hurn, J.~Jeisman, and K.~Lindsay.
\newblock Teaching an old dog new tricks: Improved estimation of the parameters
  of stochastic differential equations by numerical solution of the
  {F}okker-{P}lanck equation.
\newblock NCER Working Paper Series~9, National Centre for Econometric
  Research, Feb. 2007.

\bibitem{KH:10}
B.~Kaltenbacher and B.~Hofmann.
\newblock Convergence rates for the iteratively regularized {G}auss-{N}ewton
  method in {B}anach spaces.
\newblock {\em Inverse Problems}, 26(3):035007, 21, 2010.

\bibitem{KNS:08}
B.~Kaltenbacher, A.~Neubauer, and O.~Scherzer.
\newblock {\em Iterative Regularization Methods for Nonlinear ill-posed
  Problems}.
\newblock Radon Series on Computational and Applied Mathematics. de Gruyter,
  Berlin, 2008.

\bibitem{Kutoyants:04}
Y.~A. Kutoyants.
\newblock {\em Statistical inference for ergodic diffusion processes}.
\newblock Springer Series in Statistics. Springer-Verlag London Ltd., London,
  2004.

\bibitem{massart:00}
P.~Massart.
\newblock About the constants in {T}alagrand's concentration inequalities for
  empirical processes.
\newblock {\em Ann. Prob.}, 28(2):863--884, 2000.

\bibitem{MH:08}
P.~Math{\'e} and B.~Hofmann.
\newblock How general are general source conditions?
\newblock {\em Inverse Problems}, 24(1):015009, 5, 2008.

\bibitem{PPRS:12}
O.~Papaspiliopoulos, Y.~Pokern, G.~O. Roberts, and A.~M. Stuart.
\newblock Nonparametric estimation of diffusions: a differential equations
  approach.
\newblock {\em Biometrika}, 99(3):511--531, 2012.

\bibitem{PSZ:12}
Y.~Pokern, A.~M. Stuart, and J.~H. van Zanten.
\newblock Posterior consistency via precision operators for bayesian
  nonparametric drift estimation in {SDE}s.
\newblock {\em Stochastic Processes and their Applications}, 123(2):603 -- 628,
  2013.

\bibitem{resmerita:05}
E.~Resmerita.
\newblock Regularization of ill-posed problems in {B}anach spaces: convergence
  rates.
\newblock {\em Inverse Problems}, 21(4):1303--1314, 2005.

\bibitem{RA:07}
E.~Resmerita and R.~S. Anderssen.
\newblock Joint additive {K}ullback-{L}eibler residual minimization and
  regularization for linear inverse problems.
\newblock {\em Math. Methods Appl. Sci.}, 30(13):1527--1544, 2007.

\bibitem{RB:03}
P.~Reynaud-Bouret.
\newblock Adaptive estimation of the intensity of inhomogeneous {P}oisson
  processes via concentration inequalities.
\newblock {\em Probab. Theory Related Fields}, 126(1):103--153, 2003.

\bibitem{Risken:89}
H.~Risken.
\newblock {\em The {F}okker-{P}lanck equation}, volume~18 of {\em Springer
  Series in Synergetics}.
\newblock Springer-Verlag, Berlin, second edition, 1989.
\newblock Methods of solution and applications.

\bibitem{scherzer_etal:09}
O.~Scherzer, M.~Grasmair, H.~Grossauer, M.~Haltmeier, and F.~Lenzen.
\newblock {\em Variational methods in imaging}, volume 167 of {\em Applied
  Mathematical Sciences}.
\newblock Springer, New York, 2009.

\bibitem{Schmisser:13}
E.~Schmisser.
\newblock Penalized nonparametric drift estimation for a multidimensional
  diffusion process.
\newblock {\em Statistics}, 47(1):61--84, 2013.

\bibitem{Schuss:10}
Z.~Schuss.
\newblock {\em Theory and applications of stochastic processes}, volume 170 of
  {\em Applied Mathematical Sciences}.
\newblock Springer, New York, 2010.
\newblock An analytical approach.

\bibitem{SSOH:08}
A.~Singer, Z.~Schuss, A.~Osipov, and D.~Holcman.
\newblock Partially reflected diffusion.
\newblock {\em SIAM J. Appl. Math.}, 68(3):844--868, 2007/08.

\bibitem{Spokoiny:00}
V.~G. Spokoiny.
\newblock Adaptive drift estimation for nonparametric diffusion model.
\newblock {\em Ann. Statist.}, 28(3):815--836, 2000.

\bibitem{Talagrand:96}
M.~Talagrand.
\newblock New concentration inequalities in product spaces.
\newblock {\em Invent. Math.}, 126(3):505--563, 1996.

\bibitem{WH:12}
F.~Werner and T.~Hohage.
\newblock Convergence rates in expectation for {T}ikhonov-type regularization
  of inverse problems with {P}oisson data.
\newblock {\em Inverse Problems}, 28(10):104004, 15, 2012.

\end{thebibliography}
\bibliographystyle{abbrv}

\end{document}